\documentclass[11pt]{article}
\usepackage[margin=1in]{geometry}
\usepackage{rpmacros}
\usepackage[T1]{fontenc}
\usepackage{gitinfo}
\usepackage[usenames,dvipsnames]{xcolor}
\usepackage{stmaryrd}
\usepackage{xfrac}
\usepackage{mathpazo}
\usepackage{bm}
\usepackage{todonotes}
\usepackage{lipsum}
\usepackage{setspace}
\usepackage{scrextend}
\usepackage[shortlabels]{enumitem}
\usepackage{float}

\usepackage[linesnumbered,vlined,ruled]{algorithm2e}

\RequirePackage[colorlinks=true]{hyperref}
\hypersetup{
    linkcolor=[rgb]{0.3,0.3,0.6},
    citecolor=[rgb]{0.2, 0.6, 0.2},
    urlcolor=[rgb]{0.6, 0.2, 0.2}
}

\usepackage{amsthm}
\usepackage{thmtools,thm-restate}
\usepackage[nameinlink,capitalise]{cleveref}

\numberwithin{equation}{section}
\declaretheoremstyle[bodyfont=\it,qed=\qedsymbol]{noproofstyle}

\declaretheorem[numberlike=equation]{observation}
\declaretheorem[numberlike=equation,style=noproofstyle,name=Observation]{observationwp}
\declaretheorem[name=Observation,numbered=no]{observation*}

\declaretheorem[numberlike=equation]{problem}

\declaretheorem[numberlike=equation]{theorem}

\declaretheorem[name=Theorem,numbered=no]{theorem*}

\declaretheorem[numberlike=equation]{lemma}
\declaretheorem[name=Lemma,numbered=no]{lemma*}
\declaretheorem[numberlike=equation,style=noproofstyle,name=Lemma]{lemmawp}

\declaretheorem[numberlike=equation]{corollary}
\declaretheorem[name=Corollary,numbered=no]{corollary*}

\declaretheorem[numberlike=equation]{proposition}
\declaretheorem[name=Proposition,numbered=no]{proposition*}

\declaretheorem[name=Claim,numbered=no]{claim*}

\declaretheorem[name=Conjecture,numbered=no]{conjecture*}

\declaretheorem[name=Question,numbered=no]{question*}

\declaretheoremstyle[bodyfont=\it,qed=$\lozenge$]{defstyle} 

\declaretheorem[numberlike=equation,style=defstyle]{definition}
\declaretheorem[unnumbered,name=Definition,style=defstyle]{definition*}

\declaretheorem[unnumbered,name=Example,style=defstyle]{example*}

\declaretheorem[unnumbered,name=Notation=defstyle]{notation*}

\declaretheorem[unnumbered,name=Construction,style=defstyle]{construction*}

\declaretheorem[numberlike=equation,style=defstyle]{remark}
\declaretheorem[unnumbered,name=Remark,style=defstyle]{remark*}

\renewcommand{\phi}{\varphi}
\renewcommand{\epsilon}{\varepsilon}

\newcommand{\C}{\mathbb{C}}

\usepackage{nth}
\usepackage{intcalc}
\usepackage{etoolbox}
\usepackage{xstring}

\usepackage{ifpdf}
\ifpdf
\else
\usepackage[quadpoints=false]{hypdvips}
\fi

\newcommand{\ECCCurl}[2]{http://eccc.hpi-web.de/report/\ifnumcomp{#1}{>}{93}{19}{20}#1/#2/}

\newcommand{\shortECCC}[2]{\texttt{\href{http://eccc.hpi-web.de/report/\ifnumcomp{#1}{>}{93}{19}{20}#1/#2/}{eccc:TR#1-#2}}}

\newcommand{\parseECCC}[1]{
\StrSubstitute{#1}{TR}{}[\tmpstring]%
\IfSubStr{\tmpstring}{/}{ 
\StrBefore{\tmpstring}{/}[\ecccyear]%
\StrBehind{\tmpstring}{/}[\ecccreport]%
}{
\StrBefore{\tmpstring}{-}[\ecccyear]%
\StrBehind{\tmpstring}{-}[\ecccreport]%
}%
\shortECCC{\ecccyear}{\ecccreport}}

\newif\ifnote
\notefalse
\ifnote
\newcommand{\phnote}[1]{\todo[color=red!100!green!33,size=\footnotesize]{ph: #1}}
\newcommand{\sriprahladhuvacha}[1]{\todo[color=red!100!green!33,inline,size=\small]{ph: #1}}
\newcommand{\RPnote}[1]{\textcolor{BrickRed}{\guillemotleft RP: #1 \guillemotright}}
\newcommand{\MKnote}[1]{\textcolor{Orange}{\guillemotleft MK: #1 \guillemotright}}
\newcommand{\SGnote}[1]{\textcolor{Blue}{\guillemotleft SG: #1 \guillemotright}}
\newcommand{\gitinfonotecolour}{Gray}
\else
\newcommand{\phnote}[1]{}
\newcommand{\sriprahladhuvacha}[1]{}
\newcommand{\RPnote}[1]{}
\newcommand{\MKnote}[1]{}
\newcommand{\SGnote}[1]{}
\newcommand{\gitinfonotecolour}{white}
\fi

\newcommand{\ignore}[1]{}
\newcommand{\gitinfonote}{git info:~\gitAbbrevHash\;,\;(\gitAuthorIsoDate)\; \;\gitVtag}
\newcommand{\ehref}[1]{\href{mailto:#1}{#1}}

\newcommand{\round}[2]{{\lfloor #1 \rceil}_{#2}}
\newcommand*\samethanks[1][\value{footnote}]{\footnotemark[#1]}

\onehalfspace
\allowdisplaybreaks

\title{Fast Numerical Multivariate Multipoint Evaluation}
\author{{Sumanta Ghosh \thanks{California Institute of Technology, Pasadena, USA. \ehref{besusumanta@gmail.com}}}
\and 
{Prahladh Harsha\thanks{Tata Institute of Fundamental Research, Mumbai, India. \ehref{\{prahladh, mrinal, ramprasad\}@tifr.res.in}. Research supported by the Department of Atomic Energy, Government of India, under project 12-R\&D-TFR-5.01-0500. Research of second author partially supported by Google India Research Award.}}
\and 
{Simao Herdade\thanks{Yahoo Research, San Francisco, USA. \ehref{simaoherdade@gmail.com}} }
\and 
{Mrinal Kumar{\samethanks[2]}} 
\and 
{Ramprasad Saptharishi{\samethanks[2]}}
}
\date{}

\begin{document}
\maketitle

{\let\thefootnote\relax
\footnotetext{\textcolor{\gitinfonotecolour}{\gitinfonote (nothing to see here)}
}}


\begin{abstract}
We design nearly-linear time numerical algorithms for the problem of multivariate multipoint evaluation over the fields of rational, real and complex numbers. We consider both \emph{exact} and \emph{approximate} versions of the algorithm. The input to the algorithms are (1) coefficients of an $m$-variate polynomial $f$ with degree $d$ in each variable, and (2) points $\veca_1,\ldots, \veca_N$ each of whose coordinate has value bounded by one and bit-complexity $s$. 

\begin{description}
    \item{Approximate version:} Given additionally an accuracy parameter $t$, the algorithm computes rational numbers $\beta_1,\ldots, \beta_N$ such that $\abs{f(\veca_i) - \beta_i} \leq \sfrac{1}{2^t}$ for all $i$, and has a running time of $\inparen{(Nm + d^m)(s + t)}^{1 + o(1)}$ for all $m$ and all sufficiently large $d$. 
    \item{Exact version (when over rationals):} Given additionally a bound $c$ on the bit-complexity of all evaluations, the algorithm computes the rational numbers $f(\veca_1), \ldots, f(\veca_N)$, in time $\inparen{(Nm + d^m)(s + c)}^{1 + o(1)}$ for all $m$ and all sufficiently large $d$. . 
\end{description}

Our results also naturally extend to the case when the input is over the field of real or complex numbers under an appropriate standard model of representation of field elements in such fields. 

Prior to this work, a nearly-linear time algorithm for multivariate multipoint evaluation (exact or approximate) over any infinite field appears to be known only for the case of univariate polynomials, and was discovered in a recent work of Moroz \cite{Moroz2021}. In this work, we extend this result from the univariate to the multivariate setting. However, our algorithm is based on ideas that seem to be conceptually different from those of Moroz \cite{Moroz2021} and crucially relies on a recent algorithm of Bhargava, Ghosh, Guo, Kumar \& Umans \cite{BhargavaGGKU2022} for multivariate multipoint evaluation over finite fields, and known efficient algorithms for the problems of rational number reconstruction and fast Chinese remaindering in computational number theory. 
\end{abstract}

\section{Introduction}\label{sec: intro}

In this paper, we study the problem of designing fast algorithms for the following natural computational problem. 
\begin{quote}
    Given an $m$ variate polynomial $f$ of degree less than $d$ in each variable over an underlying field $\F$ as a list of coefficients, and (arbitrary) evaluation points $\veca_1, \veca_2, \ldots, \veca_N \in \F^m$, output $f(\veca_i)$ for every $i$.
\end{quote}

This computational task is referred to as \emph{Multivariate Multipoint Evaluation (MME)} in literature and fast algorithms for MME are of fundamental interest in computational algebra, not only due to the evident natural appeal of the problem but also due to potential applications of MME as an important subroutine for algorithms for many other algebraic problems (see \cite{KedlayaU2011} for a detailed discussion on these applications). 

The input for MME is specified by $(d^m + Nm)$ field elements, or alternatively $(d^m + Nm) \cdot s$ bits, where $s$ is an upper bound on the bit complexity of any field constant in the input. For finite fields, $s$ can be taken to be $\log |\F|$. Clearly, there is an algorithm for this problem that takes roughly $(d^m \cdot Nm)^{1 + o(1)}$  many field operations or about $(d^m \cdot Nm \cdot s)^{1 + o(1)}$ bit operations: we iteratively evaluate the polynomial one input point at a time. Obtaining significantly faster and ideally \emph{nearly-linear}\footnote{We say that an algorithm has time complexity nearly-linear in the input size if for all sufficiently large $n$, the algorithms runs on inputs of size $n$ in time $n^{1 + o(1)}$.} time algorithms for MME is the main question motivating this work. Here the time complexity of an algorithm could be measured either in terms of the number of field operations (in case the algorithm is ``algebraic\footnote{Algorithms for MME that only need arithmetic over the underlying field in their execution, or in other words can be modelled as an arithmetic circuit over the underlying field are  referred to as algebraic.}'' in the sense that only uses field operations over the underlying field, e.g. like the trivial algorithm outlined above) or the number of bit operations.

\subsection{Prior work}
Before describing the precise problem studied in this work and our main results, we start with a brief discussion of the current state of art of algorithms for MME. While the results in this paper are over infinite fields like reals, rationals and complex numbers, we begin our discussion of prior work on MME by recalling the state of affairs over finite fields. 

\subsubsection{Multipoint evaluation over finite fields} 
Multipoint evaluation of polynomials is a non-trivial problem even for the case of univariate polynomials, and a non-trivial algorithm is unclear even for this case over any (sufficiently large) field. When the set of input points have additional structure, for instance, they are all roots of unity of some order over the underlying field, the Fast Fourier Transform (FFT) gives us a nearly-linear time algorithm for this problem. However, it is not immediately clear whether ideas based on FFT can be easily extended to the case of arbitrary evaluation points. 

In a beautiful work in 1974, Borodin and Moenck \cite{BorodinM1974} designed a significantly faster algorithm for univariate multipoint evaluation by building on FFT and a fast algorithm for division with remainder for univariate polynomials. The algorithm of Borodin and Moenck worked over all fields and was algebraic, in the sense mentioned earlier, the number of field operations needed by the algorithm was $(N + d)^{1 + o(1)}$, nearly-linear in the number of field elements in the input. 

Extending these fast algorithms for multipoint evaluation from the univariate to the multivariate case proved to be quite challenging, even for the bivariate case. N\"usken and Ziegler \cite{NuskenZ2004} gave a non-trivially fast algorithm for this problem over all fields, although the precise time complexity of their algorithm was not nearly linear in the input size. The state of art for this problem saw a significant improvement in the works of Umans \cite{Umans2008} and Kedlaya \& Umans \cite{KedlayaU2008} (see also \cite{KedlayaU2011}) who gave fast algorithms for MME for the case when the number of variables $m$ is significantly smaller than the degree parameter $d$, i.e. $m = d^{o(1)}$, over fields of small characteristic and all finite fields respectively.  

This case of large number of variables was addressed recently in works of Bhargava, Ghosh, Kumar \& Mohapatra \cite{BhargavaGKM2022} and Bhargava, Ghosh, Guo, Kumar \& Umans \cite{BhargavaGGKU2022} who gave fast\footnote{Strictly speaking, these algorithms do not run in nearly-linear time, since the running time has $(\log |\F|)^c$ factor where $c$ is a fixed constant that can be greater than one. However, the dependence of the running time on the term $(d^m + Nm)$ is nearly-linear. } algorithms for MME over fields of small characteristic and over all finite fields respectively, for all sufficiently large $m, d$. 

We also note that the algorithms of Kedlaya \& Umans \cite{KedlayaU2008} and those of Bhargava, Ghosh, Guo, Kumar \& Umans \cite{BhargavaGGKU2022} for MME over all finite fields are not algebraic, and in particular rely on bit access to the input field elements and bit operations on them. This is in contrast to the algorithms of Umans \cite{Umans2008} and  Bhargava, Ghosh, Kumar \& Mohapatra \cite{BhargavaGKM2022} for MME over finite fields of small characteristic that are algebraic in nature. Designing algebraic algorithms for MME over all finite fields continues to be a very interesting open problem in this line of research. 
`
\subsubsection{Multipoint evaluation over infinite fields}

As we shall see, our understanding of the problem here is rather limited compared to that over finite fields. However, before moving on to the results, we first discuss some subtleties with the definition of this problem itself over infinite fields.

\paragraph{Field operations vs bit complexity:}

Field arithmetic over finite fields preserves the worst case bit complexity of the constants generated, but this is not the case over infinite fields. This increase in bit-complexity in intermediate computations leads to some issues that we discuss next. 

The first issue is that even the bit complexity of the output may not be nearly-linear in the input bit complexity, thereby ruling out any hope of having an algorithm with time complexity nearly-linear in the bit complexity of the input. The second issue is that even for inputs where the bit complexity of the input field elements and the output field elements are promised to be small, it might be the case that in some intermediate stage of its run, an algorithm for MME generates field elements of significantly large bit complexity. For instance, the classical algorithm of Borodin \& Moenck for univariate multipoint evaluation has near linear complexity in terms of the number of field operations, but it is not clear if the bit complexity of the algorithm is also nearly-linear in the input and output bit complexities. 

\paragraph{The input and output model:}

\SGnote{PH had a comment about decimal vs bit. Not sure what to do.}
For fields such as real or complex numbers, we need to specify a model for providing the inputs which potentially require infinite precision. The standard model used in numerical algorithms is via black-boxes that we refer to as \emph{approximation oracles} (formally defined in \cref{def:approx oracles}). Informally an approximation oracle for a real number $\alpha \in (-1, 1)$ provides, for every $k \in \N$, access to the first $k$ bits of $\alpha$ after the decimal, and its sign in time $\tilde{O}(k)$ (for complex numbers, we will assume the real and imaginary parts are provided via such oracles). 

For the output, we could either ask to compute the evaluations to the required precision, or compute the evaluations exactly when, say, in the case of rational numbers. In this paper, we consider both versions of these problems. \\

Note that computing a real number $\alpha \in (-1,1)$ within a given error $\epsilon < 1$ is essentially the same as computing the most significant $\Omega(\log \sfrac{1}{\epsilon})$ bits of the output correctly. In this sense, $O(\log \sfrac{1}{\epsilon})$ provides a natural upper bound on the bit complexity of the output for an instance of approximate multipoint evaluation. Perhaps a bit surprisingly, we did not know an algorithm for  multipoint evaluation with bit complexity nearly-linear in input size and $(\log \sfrac{1}{\epsilon})$ even for the setting of univariate polynomials till very recently. This is in contrast to the result of Borodin \& Moenck~\cite{BorodinM1974} that obtains an upper bound on the number of field operations (but not the number of bit operations) for (exact) univariate multipoint evaluation over all fields. 

In a beautiful recent work, Moroz \cite{Moroz2021} designed such an algorithm for the approximation version of univariate multipoint evaluation. Formally, he proved the following theorem. 

\begin{theorem}[Moroz \cite{Moroz2021}]\label{thm: moroz}
    There is a deterministic algorithm that takes as input a univariate polynomial $f(x) = \sum_{i = 0}^d f_ix^i \in \C[x]$ as a list of complex coefficients, with $(\abs{f}_1 := \sum_{i = 0}^d \abs{f_i} \leq 2^{\tau})$ and inputs $a_1, a_2, \ldots, a_d \in \C$ of absolute value less than one, and outputs $\beta_1, \beta_2, \ldots, \beta_d \in \C$ such that for every $i$, 
    \[
    \abs{{f(a_i) - \beta_i}} \leq \abs{f}_1 \cdot 2^{-t} \, ,
    \]
    and has bit complexity at most $\Tilde{O}(d(\tau + t))$. 
\end{theorem}

As our main result in this paper, we prove a generalization of \cref{thm: moroz} to the multivariate setting. 

\subsection{Our results} \label{sec: results}
Before stating our results, we formally define the problems that we study. The first question of approximate-MME is essentially the generalization of the univariate version of the problem studied by Moroz \cite{Moroz2021}. For convenience, we state the problem for the fields of rational and real numbers, but they extend in a straightforward manner to complex numbers and other natural subfields of it.

\begin{problem}[Approximate multivariate multipoint evaluation (approximate-MME)]\label{problem:approx-mme}
    We are given as input an $m$-variate polynomial $f \in \R[\vecx]$ of degree at most $(d-1)$ in each variable as a list of coefficients, points $\veca_1,\ldots, \veca_N \in \R^m$, and an accuracy parameter $t \in \N$. Here every field element is assumed to be in $(-1,1)$ and is given via an approximation oracle. 
    
    Compute rational numbers $\beta_1, \ldots, \beta_N$ such that $\abs{f(\veca_i) - \beta_i} < 1/2^t$ for all $i \in [N]$. 
\end{problem}

\noindent
We also study the following variant of MME in the paper. 
\begin{problem}[Exact multivariate multipoint evaluation (exact-MME)]\label{problem:exact-mme}
    We are given as input an $m$-variate polynomial $f \in \Q[\vecx]$ of degree at most $(d-1)$ in each variable as a list of coefficients, points $\veca_1,\ldots, \veca_N \in \Q^m$ and an integer parameter $s>0$, such that all rational numbers in the input and \emph{output} are expressible in the form $p/q$ for integers $p, q$ with $|p|, |q| < 2^s$ and every rational number in the input has absolute value less than one.  
    
    Compute $f(\veca_1), \ldots, f(\veca_N)$. 
\end{problem} 


The restriction that the absolute value of all constants is at most one requires a short discussion. The restriction on the coefficients of $f$ is without loss of generality (by scaling) but the restriction on the coordinates of points is \emph{not} without loss of generality, but is nevertheless well-motived. See \cref{remark:absolute-value-constants} for details. \\

\noindent
Our main result is fast algorithms for \cref{problem:approx-mme} and \cref{problem:exact-mme} for all sufficiently large $d$.  

\begin{theorem}[Approximate multipoint evaluation - informal] \label{thm: approx mme informal}
There is a deterministic algorithm for approximate-MME (\cref{problem:approx-mme}) that runs in time 
\[
\left((Nm + d^m)t\right)^{1 + o(1)}\, ,
\]
for all sufficiently large $d, t$ and all $m$. 
\end{theorem}


\begin{theorem}[Exact multipoint evaluation - informal] \label{thm: exact mme informal}
There is a deterministic algorithm for exact-MME over rational numbers (\cref{problem:exact-mme}) that runs in time 
\[
\left((Nm + d^m)s\right)^{1 + o(1)} \, 
\]
for all sufficiently large $d, s$ and all $m$. 
\end{theorem}


\cref{thm: approx mme informal} is a generalization (by scaling coefficients) of \cref{thm: moroz} of Moroz in the sense that it handles an arbitrarily large number of variables. Interestingly, our proof is \emph{not} an extension of the proof of \cref{thm: moroz} to larger number of variables.  It relies on a different set of ideas and appears to be conceptually different from the proof of Moroz \cite{Moroz2021}. Moroz's algorithm relies on geometric ideas, and does not involve any modular arithmetic, whereas ours crucially relies on various reductions from an instance of MME (approximate or exact) over rational, real or complex numbers to instances of MME over finite fields. In fact, a generalization of Moroz's univariate algorithm to higher dimensions is not immediately clear to us, and would be interesting to understand. 

As discussed in the introduction, while measuring the complexity of algorithms for MME over the field of rational numbers in terms the number of bit operations, the dependence of the running time on the bit complexity of the output, as in \cref{thm: exact mme informal} is quite natural and essentially unavoidable. However, the fact that \cref{thm: exact mme informal} takes the bit complexity of the output  as a part of its input does not seem very natural and desirable. It would be very interesting to have an algorithm for exact-MME over rationals that does not need a bound on the output complexity as a part of the input, but runs in time nearly-linear in the input and output bit complexity. 

\subsection{Overview of the proofs} \label{sec: proof overview}
In this section, we outline the main ideas in the proofs of  \cref{thm: approx mme informal} and \cref{thm: exact mme informal}. For this discussion, we assume for simplicity that the input is over the field of rational numbers, and the field constants in the input are given to us exactly. The ideas here generalize to the setting of real inputs (for approximate-MME) by some clean and simple properties of approximation oracles.

\subsubsection{A \naive{} first attempt}

We start by setting up some notation. Let $f \in \Q[\vecx]$ be an $m$ variate polynomial of degree at most $(d-1)$ in each variable and let $\veca_1, \veca_2, \ldots, \veca_N \in \Q^m$ be the input points of interest. For now, let us assume that our goal is to output the value of $f$ on each $\veca_i$ exactly. We are also given the positive integer $t$ such that the numerator and the denominator of each of the field constants in the input, and the output are at most $2^t$.

From the prior work of Bhargava, Ghosh, Guo, Kumar and Umans \cite{BhargavaGGKU2022} we have fast algorithms for MME over all finite fields. Therefore, a natural strategy for solving MME over rational numbers is to somehow reduce the problem over rationals to instances of the problem over finite fields, and use the known algorithms for this problem over finite fields to solve these instances. A first step towards such a reduction would be to clear all the denominators in the input instance by taking their LCM  and obtain an instance of MME over the ring of integers, and then work modulo a large enough prime (or several small enough primes if needed for efficiency reasons), thereby reducing to instances of MME over finite fields. However, this seemingly natural approach runs into fundamental issues even for the simpler setting where each evaluation point $\veca_i$ has integer coordinates, and the only rational numbers appear in the coefficients of the polynomial $f$. We now briefly elaborate on this issue. 

Let us consider an input instance where every denominator in the coefficient vector of $f$ is a distinct prime. For instance, we can get such an instance where each of the first $d^m$ primes appears as a denominator of some coefficient of $f$. Note that the input bit complexity parameter $t$ needs to be at most $\poly(\log d, m)$ for this case. Since the length of this coefficient vector is $d^m$,  the LCM of these denominators is a natural number that is at least as large as the product of the first $d^m$ primes, which is at least $2^{d^m}$, and hence has bit complexity at least $d^m$. Thus, if we clear out the denominators of the coefficients of $f$ to obtain a polynomial $\hat{f}$ with integer coefficients, each of the coefficients of $\hat{f}$ can have bit complexity as large as $d^m$. In this case, the total bit complexity of the coefficient vector of $\hat{f}$ is at least $d^{2m}$, which is roughly quadratic in the original input size, and thus, any algorithm obtained via this approach will have prohibitively large time complexity.

In both our algorithms for approximate-MME and exact-MME, we indeed crucially rely on the algorithms for MME over finite fields due to Bhargava et al \cite{BhargavaGGKU2022}. However, this reduction is somewhat delicate and needs some care. On the way we rely on some well known tools from  computational algebra and number theory, like fast Chinese remaindering, fast rational reconstruction, Kronecker and inverse Kronecker maps. Perhaps a bit surprisingly, our algorithm for exact-MME uses the algorithm for approximate-MME as a subroutine. 

\begin{figure}[H]
    \caption{Overview of reductions}
    \label{fig:reductions}
    \begin{center}
        \begin{tikzpicture}
            \draw (-2,-1) rectangle (2,1) node [align=center, pos = 0.5] {Exact-MME \\over integers\\ {\tiny Fast-CRT + Faster \cite{BhargavaGGKU2022}}};
            \draw[->, ultra thick] (-5.5, 0) -- (-2,0);
            \draw[->, ultra thick] (2, 0) -- (6.5,0);

            \node[font=\footnotesize, align=center, anchor=center] at (-3,0.5) {truncate \\ + scale};
            \node[font=\footnotesize, align=center, anchor=center] at (3,0.5) {un-scale};
            \node[font=\footnotesize, align=center, anchor=center] at (5,0.75) {Fast\\continued\\fractions};

            \draw (-4,-1.5) rectangle (4,2);
            \node at (0,1.7) {Approximate-MME};

            \draw (-5,-2) rectangle (6,3);
            \node at (0,2.7) {Exact-MME over rationals};
        \end{tikzpicture}
    \end{center}
\end{figure}
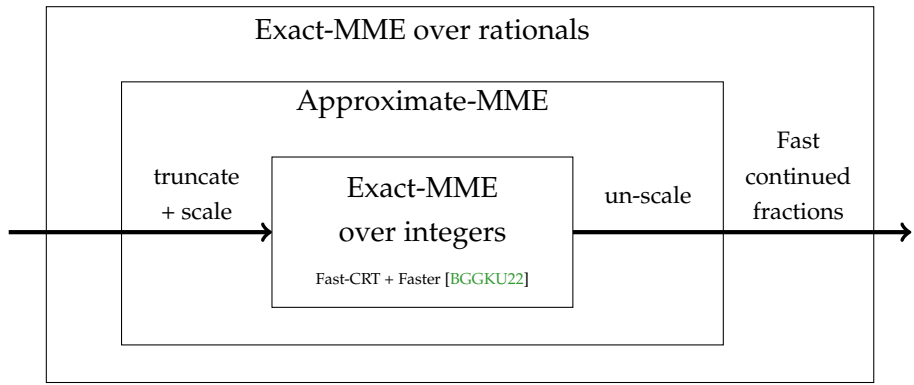

We now give an overview of the main ideas in these algorithms. We start with a very simple algorithm for exact-MME for the special case of  integer inputs that serves as a crucial subroutine for the algorithm for approximate-MME.   

\subsubsection{Algorithm for exact-MME over integers}
For this problem, all the field elements in the input are integers and the absolute value of each of these input field elements and those in the output is at most $2^s$ for a given parameter $s$.

The algorithm for MME for this case simply does this computation by working modulo a large enough prime (based on the given input and output complexities), thereby giving us a reduction from the problem over integers to that over a large enough finite field. At this point, we essentially invoke the algorithm of Bhargava et al for MME over finite fields to solve this problem. One subtlety here is that as stated in their paper \cite{BhargavaGGKU2022}, the algorithm does not quite run in nearly-linear time due to two factors. The first issue is that the running time has a $\poly(\log |\F|)$ term, where the degree of $\poly()$ term can be strictly larger than one. The other issue is that even in terms of the dependence on $d, m$, their algorithm is nearly-linear time only when $m$ is growing. So, for constant $m$, we cannot directly invoke the algorithm in \cite{BhargavaGGKU2022} for our applications. 

We get around both these issues using some simple ideas. To address the issue of a constant number of variables, we artificially increase the number variables, while reducing the individual degree bound by applying an inverse-Kronecker map to the polynomial. Then, to deal with the issue of  dependence of the running time on the field size, we first do a lift to integers and a Chinese remaindering to reduce this problem to many such instances of MME over smaller finite fields. This is essentially the same as the reduction used by Kedlaya \& Umans in \cite{KedlayaU2011}. To keep the running time nearly-linear, we do the Chinese remaindering using the well known nearly-linear time algorithms. The details can be found in \cref{sec:optimization over finite fields} and \cref{sec:exact-MME-integers}. 

\subsubsection{Algorithm for approximate-MME}

Recall that for approximate-MME, we do not need to compute the value of the polynomial on the input points exactly, but only require the outputs to be within some error of the actual evaluations. For simplicity, let us assume that the input polynomial and the evaluation points are all rational numbers, and are given exactly. As alluded to earlier in this section, it seems difficult to simply clear the denominators (via an LCM) and reduce to the integer case since there are instances, like when the denominators are all distinct primes, where this process prohibitively blows up the size of the coefficients. However, working with approximations gives us the necessary wiggle room to make something close to this idea work. 

As the first step of the algorithm, we approximate all the field constants, the coefficients of the given polynomial as well as the coordinates of the input points by truncating their decimal representation to $k$ bits after decimal (for some sufficiently large $k$ to be chosen carefully). Rounding a real number $\alpha$ of absolute value at most $1$ like this gives us a rational number $\hat{\alpha}$ of the form $\sfrac{a}{2^k}$ for some integer $a$ with $|a| \leq 2^k$. Moreover, we have that $\abs{\alpha - \hat{\alpha}} < \sfrac{1}{2^k}$. We now solve MME on this instance obtained after rounding. The crucial advantage now is that since all the denominators in this rounded instance are $2^k$, their LCM is just $2^k$, and clearing the denominator no longer incurs a prohibitive increase in the bit complexity. We now invoke the algorithm for exact-MME for integer instances described in the earlier subsection. The details can be found in \cref{sec:approx mme}. 

\subsubsection{Algorithm for exact-MME over rationals}\label{sec: exact MME overview}
For our algorithm for exact-MME, we start by first invoking the algorithm for approximate-MME on the instance for a sufficiently good accuracy parameter $t$. The choice of $t$ depends upon the output bit complexity that is given to us as a part of the input. From the guarantees on the output of the approximate-MME algorithm, we know that the approximate-MME outputs rational numbers that are at most $\sfrac{1}{2^t}$ away from the true evaluations. If we can somehow recover the true evaluations from these approximations, we would be done! What we have here are instances of the following problem: our goal is to find a hidden rational number, denoted by $\sfrac{a}{b}$ (the true evaluation) and we have access to another rational number, denoted by $A/B$ (an approximation to the true evaluation), with the guarantee that $\abs{\sfrac{A}{B} - \sfrac{a}{b}} < \sfrac{1}{2^t}$ and $|A|, |B| < 2^{O(t)}$. Crucially, we also have a parameter $s$ (given to us as a part of the input) and a guarantee that $|a|, |b| < 2^s$. 

This is essentially an instance of rational number reconstruction, which is a well-studied and classical problem of interest in computational algebra and number theory. We rely on these results (essentially in a black-box manner), and in particular the notion and properties of continued fractions to solve this problem efficiently. We observe that our choice of the parameter $t$ (as a function of $s$)  implies that $\sfrac{a}{b}$ is a \emph{convergent} (a rational number obtained by a truncation of the continued fraction representation of $\sfrac{A}{B}$). This observation along with some of the properties of convergents lets us find $\sfrac{a}{b}$ in nearly-linear time given $\sfrac{A}{B}$. The details can be found in \cref{sec:exact mme rationals}.  

\section{Preliminaries}\label{sec:prelims}

\subsection*{Notation}
\begin{itemize}
    \item We will use boldface letters $\veca, \vecb$ etc. for finite-dimensional vectors. We will also use this to denote tuples of variables $\vecx = (x_1,\ldots, x_n)$ etc. Usually the dimension of the vectors would be clear from context. 
    \item For \emph{exponent vectors} $\vece = (e_1,\ldots, e_n) \in \Z_{\geq 0}^n$ and a vector $\vecx = (x_1,\ldots, x_n)$, we will use $\vecx^\vece$ to denote the monomial $x_1^{e_1} \cdots x_n^{e_n}$. 
    \item For a real number $\alpha$, we use $\round{\alpha}{}$ to denote the closest integer to $\alpha$. When $\alpha=a+\frac{1}{2}$ for some integer $a$, $\round{\alpha}{}$ is defined as $a$.
\end{itemize}

\subsection{Useful inequalities}

\begin{lemma}[Bounds on binomial series]\label{lem:bounds-binomial}
For $d \in \N$ and $\epsilon > 0$ with $|\epsilon| < \sfrac{1}{d^2}$, we have
\[
    1 + d\epsilon \leq (1 + \epsilon)^d \leq 1 + d\epsilon + d^2\epsilon^2.
\]
\end{lemma}
\begin{proof}
The inequalities are clearly true for $d = 1, 2$, so for the rest of this discussion, we assume without loss of generality that $d \geq 3$. 

For any $i \geq 3$, we have $\binom{d}{i} \leq \binom{d}{2} \cdot d^{i-2}$. Hence,
\begin{align*}
\abs{\sum_{i=3}^d \binom{d}{i} \epsilon^i} \leq \sum_{i=3}^d \binom{d}{i} \abs{\epsilon^i} \leq \binom{d}{2} \cdot \epsilon^2 \cdot \sum_{i=1}^{d-2} \abs{d\epsilon}^i < \binom{d}{2} \cdot \epsilon^2
\end{align*}
where the last inequality uses $\epsilon < \sfrac{1}{d^2}$. Therefore,
\begin{align*}
(1 + \epsilon)^d & = 1 + d\epsilon + \binom{d}{2}\epsilon^2 + \sum_{i=3}^d \binom{d}{i} \epsilon^i \geq 1 + d\epsilon
\end{align*}
and 
\begin{align*}
    (1 + \epsilon)^d & = 1 + d\epsilon + \binom{d}{2}\epsilon^2 + \sum_{i=3}^d \binom{d}{i} \epsilon^i\\
                     & \leq 1 + d\epsilon + \binom{d}{2}\epsilon^2 + \binom{d}{2}\epsilon^2\\
                     & \leq 1 + d\epsilon + d^2\epsilon^2.\qedhere
    \end{align*}
\end{proof}

\subsection{Kronecker map for base-$d$}

The Kronecker map is a commonly used tool used to perform a variable reduction without changing the underlying sparsity. This map is defined formally as follows.

\begin{definition}[Kronecker map for base-$d$]\label{defn:kronecker}
    The \emph{$c$-variate Kronecker map for base-$d$}, denoted by $\Phi_{d,m; c}$ maps $cm$-variate polynomials into a $c$-variate polynomials via
    \[
    \Phi_{d,m;c}(f(x_{1,1},\ldots, x_{1,m}, \ldots, x_{c,1}, \ldots, x_{c,m})) = f\inparen{1, y_1^d, y_1^{d^2}, \ldots, y^{d^{m-1}}, \ldots, 1, y_c^d, y_c^{d^2}, \ldots, y_c^{d^{m-1}}}.
    \]
    If $f$ is a polynomial of individual degree less than $d$, then the monomial $\vecx_1^{\vece_1} \cdots \vecx_c^{\vece_c}$ is mapped to the monomial $y_1^{e_1}\cdots y_c^{e_c}$ where $\vece_i$ is the base-$d$ representation of $e_i$. 

    In the same spirit, we define the \emph{inverse Kronecker}, denoted by $\Phi_{d,m;c}^{-1}$, that maps a $c$-variate  polynomial of individual degree less than $d^m$ into a $cm$-variate polynomial of individual degree less than $d$, given via extending the following map linearly over monomials:
    \[
        \Phi_{d,m;c}^{-1}(y_1^{e_1} \cdots y_c^{e_c}) = \vecx_1^{\vece_1} \cdots \vecx_m^{\vece_m}
    \]
    where $\vecx_i = (x_{i,1}, \ldots, x_{i,m})$ and $\vece_i \in \set{0,\ldots, d-1}^m$ is the base-$d$ representation of $e_i < d^m$. \\

    Associated with the inverse Kronecker map, we also define $\psi_{d,m;c}:\F^c \rightarrow \F^{cm}$ that acts on points, given by 
    \[\psi_{d,m;c}:(a_1,\ldots, a_c) \mapsto (1, a_1^d, \ldots, a_1^{d^{m-1}}, \ldots, 1, a_c^d, \ldots, a_c^{d^{m-1}}).\] 
\end{definition}

The inverse Kronecker map is defined so that we have the following observation.

\begin{observationwp}[Kronecker maps and evaluations]\label{obs:kronecker-evaluations}
If $f(x_1,\ldots, x_c)$ is a polynomial of individual degree less than $d^m$, then for any $\veca\in \F^c$, we have that $\Phi_{d,m;c}^{-1}(f)(\psi_{d,m;c}(\veca)) = f(\veca)$. 
\end{observationwp}

The above observation would be useful to \emph{trade-off} degree with the number of variables as needed in some of our proofs. 

\subsection{Computing all primes less than a given number}

The classical Prime Number Theorem \cite{Hadamard1896,Poussin1897} asserts that there are $\Theta(N/\log N)$ primes numbers less than $N$, asymptotically. We can compute all prime numbers less than $N$ in deterministic $\tilde{O}(N)$ time. 

\medskip

\begin{algorithm}[H]
    \caption{\texttt{PrimeSieve}}
    \label{alg:prime-sieve}
    \SetKwInOut{Input}{Input}\SetKwInOut{Output}{Output}

    \Input{An integer $N > 1$}
    \Output{All prime numbers less than $N$.}
    \BlankLine

    Initialise an array $S$ indexed with $2,3,\ldots, N$ with all values set to \textsc{True}. 

    \For{$i \leftarrow 2$ to $\sqrt{N}$}{
        \If{$S[i]$ is \textsc{True}}{

            Set $j \gets 2i$

            \While{$j \leq N$}{
                Set $S[j]$ to \textsc{False}.

                $j \gets j + i$.
            }
        }
    }
    \Return{$\setdef{i}{\text{$S[i]$ is \textsc{True}}}$.}
\end{algorithm}

\medskip

\begin{lemmawp}[Computing primes less than a given number] \label{lem:finding-prime-list}
    There is a deterministic algorithm (\cref{alg:prime-sieve}) that computes the set of all primes less $N$ in deterministic time $\tilde{O}(N)$. 
\end{lemmawp}

\subsection{Fast Chinese Remaindering}
We also rely on the following two theorems concerning fast algorithms for questions related to the Chinese Remainder Theorem (CRT). We refer the reader to the book by von zur Gathen and Gerhard \cite{GathenG-MCA} for proofs. 

\begin{lemma}[Fast-CRT: moduli computation]\label{lem:CRT-moduli-computation}
    There is an algorithm that, when given as input coprime positive integers $p_1,\ldots, p_r$ and a positive integer $N$ with $N < \prod p_i < 2^c$, computes the remainders $a_i = N \bmod p_i$ for $i = 1,\ldots, r$ in deterministic $\tilde{O}(c)$ time. 
\end{lemma}

For proof of the above lemma see \cite[Theorem 10.24]{GathenG-MCA}.

\begin{lemma}[Fast-CRT: reconstruction]\label{lem:CRT-reconstruction}
    There is an algorithm that, when given as input coprime positive integers $p_1,\ldots, p_r$ and $a_1,\ldots, a_r$ such that $0 \leq a_i < p_i$ outputs the unique integer $0 \leq N < \prod p_i$ such that $N = a_i \bmod{p_i}$ for $i=1,\ldots, r$ in deterministic $\tilde{O}(c)$ time where $\prod p_i < 2^c$. 
\end{lemma}

For proof of the above lemma see \cite[Theorem 10.25]{GathenG-MCA}

\subsection{Input model for arbitrary precision reals}

Throughout this section, we will assume that all real numbers that are ``inputs'' (namely the coefficients of the polynomial and the coordinates of the evaluation points) are in the range $(-1,1)$ and are provided via \emph{approximation oracles} with the following guarantees:

\begin{definition}[Approximation oracle]\label{def:approx oracles}
The \emph{approximation oracle} for $\alpha \in (-1,1)$, can provide the ``sign'' $\alpha$ in $O(1)$ time, and on input $k$ returns an integer $b_k \in [-2^k, 2^k]$ satisfying
\[
    \abs{\alpha - \sfrac{b_k}{2^k}} < \sfrac{1}{2^k}.
\]
We will use $\round{\alpha}{k}$ to refer to the fraction $\sfrac{b_k}{2^k}$ obtained from the approximation oracle. 

The running time of the approximation oracle is the time taken to output $b_k$. We will say that the approximation oracle is \emph{efficient} if the running time is $\tilde{O}(k)$. 
\end{definition}

Such efficient approximation oracles can be obtained for any ``natural'' real number from any sufficiently convergent series. For algebriac reals of the form $\sqrt{2}$ etc., the standard Taylor series is sufficient. Even for ``natural'' transcendental numbers, we may have such approximation oracles:
\begin{align*}
e & = 1 + \frac{1}{1!} + \frac{1}{2!} + \cdots,\\
\pi & = 4 \cdot \tan^{-1}(1)\\
    & = 4 \cdot \inparen{\tan^{-1}(\sfrac{1}{2}) + \tan^{-1}(\sfrac{1}{3})}\\
    & = 4 \cdot \Bigl(\sfrac{1}{2} - \frac{\sfrac{1}{2^3}}{3} + \frac{\sfrac{1}{2^5}}{5} - \cdots \quad\quad+\quad\quad \sfrac{1}{3} - \frac{\sfrac{1}{3^3}}{3} + \frac{\sfrac{1}{3^5}}{5} - \cdots \Bigr).
\end{align*}

Any explicit series with $\tilde{O}(k)$ terms of the series having an error less than $\sfrac{1}{2^k}$ would qualify as an efficient approximation oracle for the purposes of the approximate-MME algorithm over reals. 

\SGnote{PH says 'reorder'. Not sure what he means}

\begin{lemma}[Repeated exponentiation for approximation oracles]\label{lem:approximation-oracle-exponentiation}
    Given an approximation oracle $A$ for a real number $\alpha \in (-1,1)$ with running time $T(k)$, and any positive integer $D$, we can build an approximation oracle $A^D$ for $\alpha^D$ with running time $T(k + O(\log D)) + \tilde{O}(k \log D)$. 
\end{lemma}
\begin{proof}
    On an input $k$, we wish to find an integer $r_k \in [-2^k, 2^k]$ such that $\abs{\alpha^D - \sfrac{r_k}{2^k}} <\sfrac{1}{2^k}$. 

    Let us first consider the case when $D$ is even. Let $t = k + 3$ and suppose we recursively compute an integer $a_t \in [-2^t, 2^t]$ such that $\abs{\alpha^{D/2} - \sfrac{a_t}{2^t}} < \sfrac{1}{2^t}$. Let $\delta = \sfrac{a_t}{2^t} - \alpha^{D/2}$.
    \begin{align*}
        \abs{\alpha^D - \sfrac{a_t^2}{2^{2t}}} & = \abs{(\alpha^{D/2})^2 - (\alpha^{D/2} + \delta)^2} < 4\cdot \sfrac{1}{2^t} \leq \sfrac{1}{2^{k+1}}
    \end{align*}
    Thus, if $R_k = r_k \cdot 2^{2t - k}$ is the multiple of $2^{2t - k}$ that is closest to $a_t^2$, then
    \begin{align*}
        \abs{\alpha^D - \sfrac{r_k}{2^{k}}} & \leq \abs{\alpha^D - \sfrac{a_t^2}{2^{2t}}} + \abs{\sfrac{a_t^2}{2^{2t}} - \sfrac{r_k 2^{2t - k}}{2^{2t}}}\\
        & < \sfrac{1}{2^{k+1}} + \sfrac{2^{2t - k - 1}}{2^{2t}} \leq \sfrac{1}{2^k}.
    \end{align*}

    If $D$ is odd, then let $t = k + 4$. We use the approximation oracle $A$ to obtain an integer $b_t \in [-2^t, 2^t]$ such that $\abs{\alpha - \sfrac{b_t}{2^t}} < \sfrac{1}{2^t}$, and recursively compute an integer $a_t \in [-2^t, 2^t]$ such that $\abs{\alpha^{\sfrac{(D-1)}{2}} - \sfrac{a_t}{2^t}} < \sfrac{1}{2^t}$. Then,
    \begin{align*}
        \abs{\alpha^D - \sfrac{a_t^2 b_t}{2^{3t}}} & \leq \abs{\alpha}\abs{\inparen{\alpha^{\sfrac{(D-1)}{2}}}^2 - \sfrac{a_t^2}{2^{2t}}} + \abs{\sfrac{a_t^2}{2^{2t}}} \abs{\alpha - \sfrac{b_t}{2^t}}\\  
         & < 4 \cdot\sfrac{1}{2^t} + \sfrac{1}{2^t} \leq \sfrac{1}{2^{k+1}}.
    \end{align*}
    Similarly, if $R_t = r_t \cdot 2^{3t - k}$ is the multiple of $2^{3t - k}$ that is closest to $a_t^2 b_t$, then 
    \[
        \abs{\alpha^D - \sfrac{r_t}{2^k}} < \sfrac{1}{2^k}.
    \]

    \bigskip

    If $\mathcal{T}(k,D)$ is the running time of this algorithm (namely \cref{alg:approximation-oracle-powering}) to compute $r_k \in [-2^k, 2^k]$ such that $\abs{\alpha^D - \sfrac{r_k}{2^k}} \leq \sfrac{1}{2^k}$, then we have
    \begin{align*}
        \mathcal{T}(k,D)  & \leq \mathcal{T}(k + 4, D/2) + \tilde{O}(k)\\
                & \leq \mathcal{T}(k + O(\log D), 1) + \tilde{O}(k \log D)\\
                & = T(k + O(\log D)) + \tilde{O}(k \log D).\qedhere
    \end{align*}
\end{proof}

\begin{algorithm}
    \caption{\texttt{ApproximationOracle-Powering}}
    \label{alg:approximation-oracle-powering}
    \SetKwInOut{Input}{Input}\SetKwInOut{Output}{Output}

    \Input{An approximation oracle $A$ for a real number $\alpha$, an integer $D > 0$, and an integer $k > 0$.}
    \Output{An integer $r_k \in [-2^k,2^k]$ such that $\abs{\alpha^D - \sfrac{r_k}{2^k}} < \sfrac{1}{2^k}$.}
    \BlankLine

    \If{$D = 1$}{
        \Return{$r_k = A(k)$.}
    }
    \If{$D$ is even}{
        Let $t = k + 3$.

        Compute $a_t = \texttt{ApproximationOracle-Powering}(A, \sfrac{D}{2}, t)$.

        \Return{$\round{\sfrac{a_t^2}{2^{2t - k}}}{}$.}
    }\Else{
        Let $t = k + 4$.

        Compute $b_t = A(t)$. 
        Compute $a_t = \texttt{ApproximationOracle-Powering}(A, \sfrac{(D-1)}{2}, t)$.

        \Return{$\round{\sfrac{a_t^2 \cdot b_t}{2^{3t - k}}}{}$.}
    }
\end{algorithm}

We also note that this notion of approximation oracles naturally extends to representation of complex numbers. Here, each complex number is given by two such oracles, corresponding to the real and the imaginary part respectively. 

\section{Revisiting MME over prime fields} \label{sec:optimization over finite fields}

We recall the result of Bhargava, Ghosh, Guo, Kumar and Umans \cite{BhargavaGGKU2022}. 

\begin{theorem}[Fast multivariate multipoint evaluation over finite fields \cite{BhargavaGGKU2022}]
    \label{thm:MME-finite-fields}
    There is a deterministic algorithm that when given as input the coefficient vector of an $m$ variate polynomial $f$ of degree less than $d$ in each variable over some finite field $\F$, and $N$ points $\veca_1, \veca_2, \ldots, \veca_N \in \F^m$ outputs $f(\veca_1), f(\veca_2), \ldots, f(\veca_N)$ in time 
    \[
    (d^m + Nm)^{1 + o(1)}\cdot \poly(m, d, \log |\F|), 
    \]
    for all $m \in \N$ and sufficiently large $d \in \N$. 
\end{theorem}

The above running time is not \emph{quite} nearly-linear in the input considered as bits due to the factor of $\poly(\log|\F|)$. Also, in the setting when $m$ is a constant, we can no longer absorb $\poly(d)$ within $(Nm + d^m)^{o(1)}$. However, we show below that for the case of prime fields, we can get around these issues and obtain the following nearly linear-time bound.

\begin{theorem}[Nearly-linear time MME over prime fields] 
    \label{thm:nearly-linear-MME-finite-fields}
    There is a deterministic algorithm (namely \cref{alg:nearly-linear-mme-finite-fields}) that, when given as input the coefficient vector of an $m$-variate polynomial $f$ of degree less than $d$ in each variable over a prime field $\F_p$, and $N$ points $\veca^{(1)},\ldots, \veca^{(N)} \in \F^m$, outputs $f(\veca^{(1)}), \ldots, f(\veca^{(N)})$ in time 
    \[
    \inparen{(d^m + Nm) \cdot \log p}^{1 + o(1)}
    \]
    for all $m \in \N$ and sufficiently large $d \in \N$. 
\end{theorem}

We first discuss how we handle the two cases when the number of variables is constant and growing with the input respectively in the following two subsections and then prove \cref{thm:nearly-linear-MME-finite-fields}.

We first discuss how we handle the two cases when the number of variables is constant and growing with the input respectively in the following two subsections and then prove \cref{thm:nearly-linear-MME-finite-fields}.

\subsection{Handling cases when the number of variables is too small}

As mentioned above, in the setting when the number of variables is too small (say $m \leq c$ for a constant $c$), we may no longer have that $\poly(d) = d^{o(m)}$. However, we can use the inverse-Kronecker map (\cref{defn:kronecker}) to trade-off degree with the number of variables. 

To make the parameters more informative, we rename them and let $f$ be a $c$-variate polynomial of individual degree less than $D$, and let $\veca^{(1)}, \ldots, \veca^{(N)} \in \F_p^c$ be the points at which we wish to evaluate the polynomial. 

Let $d = \floor{\log D}$ and $m$ be the smallest integer such that $d^m > D$. Note that $d^m > D > d^{m-1}$ and $m = \Theta(\sfrac{\log D}{\log\log D})$. 
If $f(x_1,\ldots, x_c) = \sum_\vece f_\vece \cdot \vecx^\vece$, define the polynomial $g(y_{1,1}, \ldots, y_{c,m}) = \Phi_{d,m;c}^{-1}(f)$, as defined in \cref{defn:kronecker}. 

For all $i\in[N]$, define $\widetilde{\veca^{(i)}} = \psi_{d,m;c}(\veca^{(i)})$, as defined in \cref{defn:kronecker}. Then, from \cref{obs:kronecker-evaluations}, we have that $f(\veca^{(i)}) = g(\widetilde{\veca^{(i)}})$ for all $i \in [N]$. The following observation shows that $\widetilde{\veca^{(i)}}$ can be computed efficiently from $\veca^{(i)}$. 

\begin{observation}\label{obs:kronecker-ff-fast}
    Given $\veca \in \F_p^c$, the point $\widetilde{\veca} := \psi^{(c)}_{d,m}(\veca) \in \F_p^{cm}$ can be computed in $\poly(d,m, c) \cdot \tilde{O}(\log p)$ time. 
\end{observation}
\begin{proof}
The running time bound follows from repeated exponentiation as $a^{d^k}\bmod p = (a^{d^{k-1}} \bmod p)^{d} \bmod{p}$ and the fact that additions and multiplications modulo $p$ can be performed in $\tilde{O}(\log p)$ time. 
\end{proof}

Thus, the task of computing $f(\veca^{(1)}), \ldots, f(\veca^{(N)})$ reduces to the task of computing the evaluations $g(\widetilde{\veca^{(1)}}), \ldots, g(\widetilde{\veca^{(N)}})$ where $\widetilde{\veca^{(i)}} = \psi^{(c)}_{d,m}(\veca)$. Also, the reduction runs in time $((D^c + Nc) \cdot \log p)^{1 + o(1)}$ since $d,m = D^{o(1)}$. 

\subsection{When individual degree and number of variables are moderately growing}

We return to the familiar variable convention of $f(x_1,\ldots, x_m) \in \F_p[x_1,\ldots, x_m]$ with degree in each variable less than $d$. From the previous section, may assume without loss of generality that $d,m = \omega(1)$ and hence $\poly(d,m) = (d^m + Nm)^{o(1)}$. Let $f$ be written as a sum of monomials as follows.
\[
    f(x_1,\ldots, x_m) = \sum_{\vece} f_\vece \cdot x_1^{e_1} \cdots x_m^{e_m}. 
\]
Interpreting the above as a polynomial over integers with each coefficient in $\set{0,1,\ldots, p-1}$, and for any $\veca \in \set{0,\ldots, p-1}^m$, the integer $f(\veca)$ is bounded by $d^m \cdot p \cdot p^{dm}$. The idea is to use Chinese Remainder Theorem to reduce the problem to MME over smaller prime fields. 

\medskip

\begin{algorithm}[H]
    \caption{\texttt{NearlyLinearTimeMME-PrimeFields}}
    \label{alg:nearly-linear-mme-finite-fields}
    \SetKwInOut{Input}{Input}\SetKwInOut{Output}{Output}

    \Input{$f(x_1,\ldots, x_m) \in \F_p[x_1,\ldots, x_m]$ with degree in each variable less than $d$, and $\veca^{(1)},\ldots, \veca^{(N)} \in \F_p^m$.}
    \Output{Evaluations $b_i = f(\veca^{(i)})$ for $i \in [N]$.}

    \BlankLine

    \If{$m < \log\log d$}{\label{alg:nearly-linear-mme-ff:small-m-case}
        Let $d' = \floor{\log d}$ and $m'$ be the smallest integer such that $(d')^{m'} > d$. 

        Replace $f$ by $\Phi_{d',m';m}^{-1}(f)$ and each $\veca^{(i)}$ by $\psi_{d',m'; m}(\veca^{(i)})$. \label{alg:nearly-linear-mme-ff:small-m-reduction}
    }

    \BlankLine

    Let $\tilde{L} = (dm + 1)\log p + m\log d$. Compute the first $\tilde{L}$ primes numbers $\set{p_1, \ldots, p_{\tilde{L}}}$. 
    \label{alg:nearly-linear-mme-ff:sieving}

    Let $L \leq \tilde{L}$ be the smallest integer such that $p_1 \cdots p_{L} =:M > d^m \cdot p \cdot p^{dm}$. 
    \label{alg:nearly-linear-mme-ff:finding-primes}

    \For{$\vece \in \set{0,\ldots, d-1}^m$}{\label{alg:nearly-linear-mme-ff:for-crt-coeffs}
        Compute $f_{\vece}^{(\ell)} = f_\vece \bmod p_{\ell}$ for $\ell \in L$ via fast-CRT-moduli-computation (\cref{lem:CRT-moduli-computation}).  \label{alg:nearly-linear-mme-ff:crt-coeffs}
    }

    \For{$i \in [N], k\in [m]$}{\label{alg:nearly-linear-mme-ff:for-crt-points}
        Compute $a_{i,k, \ell} = \veca^{(i)}_k \bmod p_\ell$ for $\ell \in L$ via fast-CRT-moduli-computation (\cref{lem:CRT-moduli-computation}).  \label{alg:nearly-linear-mme-ff:crt-points}
    }
    \For{$\ell \in L$}{\label{alg:nearly-linear-mme-ff:for-BGGKU}
        Let $f^{(\ell)}(x_1,\ldots, x_m) = \sum_{\vece} f_{\vece}^{(\ell)} \vecx^\vece \in \F_{p_i}[\vecx]$.

        Let $\veca^{(i,\ell)} = (a_{i,1,\ell}, \ldots, a_{i,m,\ell}) \in \F_{p_\ell}^m$ for each $i \in [N]$. 

        Compute $b_{i,\ell} = f^{(\ell)}(\veca^{(i,\ell)})$ for all $i \in [N]$ using \cref{thm:MME-finite-fields}. \label{alg:nearly-linear-mme-ff:BGGKU}
    }

    \For{$i \in [N]$}{\label{alg:nearly-linear-mme-ff:for-crt-reconstruct}
        Compute the unique $b_i \in [0,M)$ such that $b_{i} = b_{i,\ell}\bmod{p_\ell}$ for all $\ell \in [L]$, via fast-CRT-reconstruction (\cref{lem:CRT-reconstruction}). \label{alg:nearly-linear-mme-ff:crt-reconstruct}
    }

    \Return{$(b_1\bmod p, \ldots, b_N \bmod p)$.}
\end{algorithm}

\begin{proof}[Proof of \cref{thm:nearly-linear-MME-finite-fields}]
    The correctness of \cref{alg:nearly-linear-mme-finite-fields} is evident.

    As for the running time, \crefrange{alg:nearly-linear-mme-ff:small-m-case}{alg:nearly-linear-mme-ff:small-m-reduction} takes $(d^m + Nm)^{1 + o(1)}$ time by \cref{obs:kronecker-ff-fast} and reduces to the case when $m \geq \log\log d$. 
    In this case, \cref{alg:nearly-linear-mme-ff:sieving,alg:nearly-linear-mme-ff:finding-primes} require $\tilde{O}(\tilde{L})$ time (\cref{lem:finding-prime-list}), which is $\tilde{O}(\log p) \cdot \poly(d,m)$. 
    
    Using \cref{lem:CRT-moduli-computation}, we have that \crefrange{alg:nearly-linear-mme-ff:for-crt-coeffs}{alg:nearly-linear-mme-ff:crt-points} require time $(d^m + Nm) \cdot \tilde{O}(\log M) = ((d^m + Nm)\cdot \log p)^{1 + o(1)}$. 

    From \cref{thm:MME-finite-fields}, we have that \cref{alg:nearly-linear-mme-ff:BGGKU} runs in time $(d^m + Nm)^{1 + o(1)} \cdot \poly(d,m, \log p_i)$, and since $p_i < \tilde{O}(\tilde{L}) = \tilde{O}(d m \log p)$, the entire loop in \crefrange{alg:nearly-linear-mme-ff:for-BGGKU}{alg:nearly-linear-mme-ff:BGGKU} takes time $(d^m + Nm)^{1 + o(1)} \cdot \tilde{O}(\log p) = ((d^m + Nm) \log p)^{1 + o(1)}$. 

    And finally, from \cref{lem:CRT-reconstruction} we have that the entire loop in \crefrange{alg:nearly-linear-mme-ff:for-crt-reconstruct}{alg:nearly-linear-mme-ff:crt-reconstruct} takes time $(Nm) \cdot \tilde{O}(\log M) = ((d^m + Nm)\cdot \log p)^{1 + o(1)}$. Hence, \cref{alg:nearly-linear-mme-finite-fields} runs in time $((d^m + Nm) \log p)^{1 + o(1)}$. 
\end{proof}

\section{Exact-MME over integers with known output  bit complexity}
\label{sec:exact-MME-integers}

In this section, we study the following version of MME over integers. 

\begin{quote}
\textbf{Input:} An integer $s > 0$, a polynomial $f(x_1,\ldots, x_m) \in \Z[x_1,\ldots, x_m]$ of individual degree less than $d$, given as a list of $d^m$ integer coefficients, a set of points $\veca^{(1)}, \ldots, \veca^{(N)} \in \Z^m$ with each coordinate of magnitude at most $2^s$, with the guarantee that all coefficients of $f$, coordinates of $\veca^{(i)}$'s, and evaluations $f(\veca^{(i)})$ are bounded in magnitude by $2^s$. 

\textbf{Output:} Integers $b_1,\ldots, b_N$ that are the evaluations, i.e. $b_i = f(\veca^{(i)})$ for $i \in [N]$.
\end{quote}
\medskip

\begin{theorem}[Exact-MME over integers]\label{thm:exact-mme-over-integers}
    There is a deterministic algorithm (namely \cref{alg:exact-mme-integers}) that on input as mentioned above returns the required output as mentioned above and runs in deterministic time $((d^m + Nm) \cdot s)^{1 + o(1)}$ for all $m \in \N$ and sufficiently large $d \in \N$. 
\end{theorem}

The main idea is to use the Chinese Remainder Theorem and reduce to the case of MME over finite fields. Since we wish to obtain a nearly-linear time algorithm, we would once again need to use Chinese Remainder Theorem implemented in nearly-linear time (\cref{lem:CRT-moduli-computation,lem:CRT-reconstruction}) and make use of the nearly-linear time algorithm for MME over prime fields (\cref{thm:nearly-linear-MME-finite-fields}). 

\medskip

\begin{algorithm}
    \caption{\texttt{ExactMME-integers}}
    \label{alg:exact-mme-integers}
    \SetKwInOut{Input}{Input}\SetKwInOut{Output}{Output}

    \Input{$f(x_1,\ldots, x_m) \in \Z[x_1,\ldots, x_m]$ and $\veca^{(1)},\ldots, \veca^{(N)} \in \Z^n$, and an integer $s > 0$ such that all coefficients of $f$, coordinates of $\veca^{(i)}$ and evaluations $f(\veca^{(i)})$ have magnitude bounded by $2^s$.}
    \Output{Evaluations $b_i = f(\veca^{(i)})$ for $i \in [N]$.}
    \BlankLine

    Compute the first $s$ primes numbers $\set{p_1, \ldots, p_s}$. \label{alg:exact-mme-integers:sieving}

    Let $L \leq s$ be the smallest integer such that $p_1 \cdots p_{L} =:M > 2^{s+1}$. \label{alg:exact-mme-integers:finding-primes}

    \For{$\vece \in \set{0,\ldots, d-1}^m$}{\label{alg:exact-mme-integers:for-crt-coeffs}
        Compute $f_{\vece}^{(\ell)} = f_\vece \bmod p_\ell$ for $\ell \in L$ using \cref{lem:CRT-moduli-computation}. \label{alg:exact-mme-integers:crt-coeffs}
    }
    \For{$i \in [N], k\in [m]$}{\label{alg:exact-mme-integers:for-crt-points}
        Compute $a_{i,k, \ell} = \veca^{(i)}_k \bmod p_\ell$ for $\ell \in L$ using \cref{lem:CRT-moduli-computation}. \label{alg:exact-mme-integers:crt-points}
    }

    \For{$\ell \in [L]$}{\label{alg:exact-mme-integers:for-ff-mme}
        Let $f^{(\ell)}(x_1,\ldots, x_m) = \sum_{\vece} f_{\vece}^{(\ell)} \vecx^\vece \in \F_{p_i}[\vecx]$.

        Let $\veca^{(i,\ell)} = (a_{i,1,\ell}, \ldots, a_{i,m,\ell}) \in \F_{p_\ell}^m$ for each $i \in [N]$. 

        Compute $b_{i,\ell} = f^{(\ell)}(\veca^{(i,\ell)})$ for all $i \in [N]$ using \cref{alg:nearly-linear-mme-finite-fields}. \label{alg:exact-mme-integers:ff-mme}
    }

    \For{$i \in [N]$}{\label{alg:exact-mme-integers:for-crt-reconstruct}
        Compute the unique $b_i \in [-\sfrac{M}{2},\sfrac{M}{2}]$ such that $b_{i} = b_{i,\ell}\bmod{p_\ell}$ for all $\ell \in [L]$, using \cref{lem:CRT-reconstruction}. \label{alg:exact-mme-integers:crt-reconstruct}
    }

    \Return{$\setdef{b_i}{i\in [N]}$.}
\end{algorithm}
\begin{proof}[Proof of \cref{thm:exact-mme-over-integers}]
    We are guaranteed that $\abs{f(\veca^{(i)})} < 2^s$ for all $i\in [N]$. Hence, by the Chinese Remainder Theorem, it is sufficient to compute $f(\veca) \bmod{p_i}$ for each $i \in [L]$ since $p_1 \cdots p_L > 2^{s + 1}$. Hence, the correctness of \cref{alg:exact-mme-integers} is evident. As for the running time, we will do an analysis very similar to the analysis for \cref{alg:nearly-linear-mme-finite-fields}. 

    Using \cref{lem:CRT-moduli-computation}, we have that \crefrange{alg:exact-mme-integers:sieving}{alg:exact-mme-integers:finding-primes} require time $O(\tilde{s})$. By the Prime Number Theorem \cite{Hadamard1896,Poussin1897}, we also have that each $p_i = \tilde{O}(s)$ and hence $p_1\cdots p_L < 2^{s+1} \cdot \tilde{O}(s)$. 

    From \cref{thm:MME-finite-fields}, we have that \cref{alg:exact-mme-integers:ff-mme} runs in time $((d^m + Nm) \cdot \log p_\ell)^{1 + o(1)}$ the entire loop in \crefrange{alg:exact-mme-integers:for-ff-mme}{alg:exact-mme-integers:ff-mme} takes time $((d^m + Nm) (\sum_\ell \log p_\ell))^{1 + o(1)} = ((d^m + Nm) \cdot s)^{1 + o(1)}$. 

    And finally, from \cref{lem:CRT-reconstruction} we have that the entire loop in \crefrange{alg:exact-mme-integers:for-crt-reconstruct}{alg:exact-mme-integers:crt-reconstruct} takes time $(Nm) \cdot \tilde{O}(\log M) = ((d^m + Nm)\cdot s)^{1 + o(1)}$. Hence, \cref{alg:exact-mme-integers} runs in time $((d^m + Nm) \cdot s)^{1 + o(1)}$ as claimed.
\end{proof}

\begin{remark}\rm
    If we are only given that all coefficients of $f$ and all coordinates of the points are integers bounded in magnitude by $2^s$ with no a-priori bound on the bit complexity of the evaluations, a \naive{} bound on the size of evaluations is 
    \[
    \abs{f(\veca)} \leq d^m \cdot 2^s \cdot 2^{sdm} \leq 2^{sdm + s + m\log d}.
    \]
    Thus, we may use $s' = (sdm + s + m\log d)$ in \cref{thm:exact-mme-over-integers} to get the time complexity bounded by $\inparen{(d^m + Nm) \cdot (sdm)}^{1 + o(1)}$. If $m$ is a growing function, then the output complexity is nearly-linear in the input complexity since $\poly(d) = (d^m + Nm)^{o(1)}$. But, in the regime when $m$ is a constant, this is super-linear in the input size $(d^m + Nm)\cdot s$ because of the additional factor of $d$. However, a slightly worse running time is to be expected in this case since the output complexity is $\Omega(N \cdot sdm)$ in the worst case. 
\end{remark}

\section{Approximate-MME over reals}\label{sec:approx mme}

Throughout this section, we will assume that all real numbers as part of the input are in the interval $(-1,1)$. 

\begin{remark}[On the restriction on absolute value of constants]
    \label{remark:absolute-value-constants}
    Given any arbitrary polynomial $f(\vecx) \in \R[\vecx]$, we can scale the polynomial by the largest coefficient to obtain and run the approximate-MME on the scaled polynomial $\tilde{f}$. If we have $\abs{\tilde{f}(\veca) - \beta_i} \leq \epsilon$, then we immediately have $\abs{f(\veca) - \inparen{\max |f_\vece| } \beta_i} \leq \epsilon \cdot \inparen{\max |f_\vece| }$. Thus, we may assume without loss of generality that all coefficients of $f$ have absolute value at most $1$. 

    However, the assumption that coordinates of all evaluation points have absolute value bounded by one is \emph{not} without loss of generality but is well-motivated nevertheless. Even in the case of univariate integer polynomials, the evaluation $f(\veca)$ could be as large as $\abs{\veca}^d$ where $\abs{\veca} = \max \abs{\veca_i}$. Therefore, the output bit-complexity for MME is potentially $O(d\cdot N)$ which is super-linear in the input bit-complexity. 

    The restriction of insisting that evaluation points consist of coordinates with absolute value at most $1$ ensures that the evaluations are never prohibitively large in magnitude, thereby making the quest for approximate-MME in nearly-linear time more meaningful.
\end{remark}

\subsection{The problem statement and algorithm}

We now state the precise problem statement and our results for approximate-MME over the field of real numbers. 

\begin{quote}
\textbf{Input:} A polynomial $f(x_1,\ldots, x_m) \in \R_{(-1,1)}[x_1,\ldots, x_m]$ of individual degree less than $d$, given as a list of $d^m$ efficient approximation oracles for each coefficient,  a set of points $\veca^{(1)}, \ldots, \veca^{(N)} \in (-1,1)^m$ each of whose coordinates are also provided via efficient approximation oracles, and an accuracy parameter $t$. 

\textbf{Output:} Rational numbers $b_1,\ldots, b_N$ such that $\abs{f(\veca^{(i)}) - b_i} < \sfrac{1}{2^t}$ for all $i \in [N]$. 
\end{quote}

\begin{theorem}[approximate-MME over reals] \label{thm:approx-MME-reals}
    There is a deterministic algorithm (namely \cref{alg:approx-mme-real}) that on input as mentioned above returns the required output as mentioned above and runs in time $((d^m + Nm) \cdot t)^{1 + o(1)}$ for all $m \in \N$ and sufficiently large $d \in \N$. 
\end{theorem}

\noindent
The rest of the section is devoted to the proof of the above theorem.

\paragraph{High-level idea:} The algorithm is a suitable reduction to the task of exact-MME over integers (\cref{thm:exact-mme-over-integers}). We will replace each of the real numbers by appropriately chosen approximations of the form $\sfrac{a_i}{2^k}$ (for a suitable large $k = O(t)$) so that the evaluations of the perturbed polynomial at the perturbed points are not too far from the original evaluations. Since we now have all denominators of the form $2^k$, we can \emph{clear} the denominators and reduce to the case of computing MME over integers. 

As expected, there are some subtleties that need to be addressed to make sure that the entire algorithm runs in nearly-linear time. 

\subsection*{Rounding coefficients of $f$}

Let $k$ be a parameter to be chosen shortly. Define the polynomial $\round{f}{k}$ as 
\[
    \round{f}{k}(x_1,\ldots, x_m) := \sum_{\vece} \round{f_\vece}{k} \cdot \vecx^{\vece}. 
\]

\begin{observation}[Error due to rounding coefficients of $f$]\label{obs:error-rounding-coeffs}
    For any $\veca \in (-1,1)^m$, we have that 
    \[
    \abs{f(\veca) - \round{f}{k}(\veca)} \leq \sfrac{1}{2^{k-m\log d}}.
    \]
\end{observation}
\begin{proof}
    \begin{align*}
        f(\veca) - \round{f}{k}(\veca) & = \sum_{\vece} (f_\vece - \round{f_\vece}{k}) \cdot \veca^\vece\\
        \implies \abs{f(\veca) - \round{f}{k}(\veca)} & \leq \sum_{\vece} \abs{f_\vece - \round{f_\vece}{k}} \cdot \abs{\veca^\vece} \leq d^m \cdot \sfrac{1}{2^k}.\qedhere
    \end{align*}
\end{proof}

\subsection*{Rounding points}

Let $k$ be a parameter to be chosen shortly. For any $\veca = (a_1,\ldots, a_m) \in (-1,1)^m$, define $\round{\veca}{k}$ as
\[
    \round{\veca}{k} := \inparen{\round{a_1}{k}, \ldots, \round{a_m}{k}}.
\]

\begin{observation}[Error due to rounding points]\label{obs:error-rounding-points}
    Let $\vece = (e_1,\ldots, e_m) \in \set{0,\ldots, d-1}^m$ and  $\veca \in (-1,1)^m$. Suppose $k \in \N$ such that $2^k > 4d^2 m^2$. Then, 
    \[
    \abs{\veca^\vece - \round{\veca}{k}^{\vece}} \leq \sfrac{1}{2^{k - \log(4dm)}}
    \]
\end{observation}
\begin{proof} Note that all $a_i \in (-1,1)$. Let $\delta_i = \round{a_i}{k} - a_i$ for $i\in [m]$; we have that  $\abs{\delta_i} \leq \sfrac{1}{2^k} \leq \sfrac{1}{4d^2m^2}$. Hence,
    \begin{align*}
        \round{a_1}{k}^{e_1} \cdots \round{a_m}{k}^{e_m} & = (a_1 + \delta_1)^{e_1} \cdots (a_m + \delta_m)^{e_m}\\
        & = a_1^{e_1} \cdots a_m^{e_m} + \sum_{\substack{j_1 \leq e_1, \ldots, j_m \leq e_m\\\text{not all $j_i = 0$}}} \binom{e_1}{j_1} \cdots \binom{e_m}{j_m} \cdot \prod_{i=1}^m \inparen{a_i^{e_i - j_i} \cdot \delta_i^{j_i}}\\
    \end{align*}
    \begin{align*}
        \implies \abs{\round{a_1}{k}^{e_1} \cdots \round{a_m}{k}^{e_m} - a_1^{e_1}\cdots a_m^{e_m}} &\leq 
        \abs{\sum_{\substack{j_1 \leq e_1, \ldots, j_m \leq e_m\\\text{not all $j_i = 0$}}} \binom{e_1}{j_1} \cdots \binom{e_m}{j_m} \cdot \prod_{i=1}^m \delta_i^{j_i}}\\
        & \leq \abs{\prod_{i=1}^m \inparen{1 + \delta_{j_i}}^d - 1}\\
        & \leq (1 + 2d(\sfrac{1}{2^k}))^m - 1 \leq 4dm(\sfrac{1}{2^k}). \quad \text{(\cref{lem:bounds-binomial})}\qedhere
    \end{align*}
\end{proof}

\subsection*{Handling the case when number of variables is too small}

To make the variables suggestive, we will rename them and say $f(x_1,\ldots, x_c)$ is a $c$-variate polynomial in $\R_{(-1,1)}[x_1,\ldots, x_c]$ with degree in each variable less than $D$. We wish to evaluate the polynomial on points $\veca^{(1)}, \ldots, \veca^{(N)} \in (-1,1)^c$. 

Once again, let $d = \floor{\log D}$ and let $m$ be the smallest integer such that $d^m > D$. Note that $d^m > D \geq d^{m-1}$ and $m = \Theta(\sfrac{\log D}{\log\log D})$. Define the polynomial $g(y_{1,1}, \ldots, y_{c,m}) = \Phi_{d,m;c}^{-1}(f)$, as defined in \cref{defn:kronecker}. Define $\widetilde{\veca^{(i)}} = \psi_{d,m;c}(\veca^{(i)})$. From \cref{obs:kronecker-evaluations}, we have that $f(\veca^{(i)}) = g(\widetilde{\veca^{(i)}})$ for all $i \in [N]$. 

Even if $\veca^{(i)}$ consisted of only rational numbers, unlike the setting in \cref{thm:nearly-linear-MME-finite-fields} where we could use \cref{obs:kronecker-ff-fast}, the rational numbers in $\widetilde{\veca^{(i)}}$ have much larger bit complexity due to the exponentiation. However, by \cref{lem:approximation-oracle-exponentiation}, we have efficient approximation oracles for $\widetilde{\veca^{(i)}}$ and that suffices for our algorithm.

\subsection{Reduction to exact-MME over integers}

From the previous subsection, we may now assume without loss of generality that we are working with an $m$-variate polynomial $f(x_1,\ldots, x_n)$ of individual degree less than $d$, with both $m,d$ as growing parameters, and wish to evaluate this polynomial on $N$ points $\veca^{(1)}, \ldots, \veca^{(N)} \in (-1,1)^m$, with all coefficients and coordinates provided via approximation oracles running in time $\tilde{O}(k + O(m \log d))$. We wish to compute integers $b_1,\ldots, b_N$ such that $\abs{f(\veca^{(i)})- \sfrac{b_i}{2^t}} < \sfrac{1}{2^t}$. We now describe the algorithm (\cref{alg:approx-mme-real}). 

\begin{algorithm}
    \caption{\texttt{approximate-MME-Reals}}
    \label{alg:approx-mme-real}
    \SetKwInOut{Input}{Input}\SetKwInOut{Output}{Output}

    \Input{An $m$-variate polynomial $f(x_1,\ldots, x_m) \in \R_{(-1,1)}[\vecx]$ of individual degree less than $d$, and points $\veca^{(1)}, \ldots, \veca^{(N)} \in \R_{(-1,1)}^m$ (with all real numbers provided via approximation oracles) and an integer $t > 0$.}
    \Output{Integers $b_1,\ldots, b_N$ such that $\abs{f(\veca^{(i)}) -  \sfrac{b_i}{2^t}} < \sfrac{1}{2^t}$ for all $i \in [N]$.}
    \BlankLine

    \If{$m < \log\log d$}{\label{alg:approx-mme-real:if-kronecker}
        Let $d' = \floor{\log d}$ and $m'$ be the smallest integer such that $(d')^{m'} > d$. 

        Replace $f$ by $\Phi_{d',m';m}^{-1}(f)$ and each $\veca^{(i)}$ by $\psi_{d',m'; m}(\veca^{(i)})$. \label{alg:approx-mme-real:kronecker}
    }

    \BlankLine

    Let $k_1 = \ceil{t + m\log d + 2}$ and $k_2 = \ceil{t + m\log d + \log(4md) + 2}$; let $k = \max(k_1,k_2) = k_2$. 

    \BlankLine

    Compute $\round{f}{k_1} = \sum_\vece \sfrac{g_{\vece,k_1}}{2^{k_1}} \cdot \vecx^{\vece} = \sfrac{1}{2^{k_1}} \cdot \sum_\vece g_{\vece,k_1} \cdot \vecx^{\vece}$. \label{alg:approx-mme-real:round-f}

    \For{$i \in [N]$}{\label{alg:approx-mme-real:for-round-points}
        Compute $\round{\veca^{(i)}}{k_2} = (\sfrac{a_{i,1,k_2}}{2^{k_2}}, \ldots, \sfrac{a_{i,m,k_2}}{2^{k_2}}) = \sfrac{1}{2^{k_2}} \cdot (a_{i,1,k_2}, \ldots, a_{i,m,k_2})$.

        Let $\widehat{\veca^{(i)}} = (a_{i,1,k_2}, \ldots, a_{i,m,k_2})$.\label{alg:approx-mme-real:round-points}
    }

    \BlankLine

    Compute the polynomial $G(x_1,\ldots, x_m)$ defined as
    \[
        G(x_1,\ldots, x_n)  = \sum_{\vece \in \set{0,\ldots,d-1}^m} g_{\vece, k_1} \cdot 2^{(k_2 d m) - k_2\abs{\vece}} \cdot \vecx^\vece
    \]
    where $\abs{\vece}$ refers to the sum of the coordinates (i.e., the degree of the monomial $\vecx^{\vece}$). \label{alg:approx-mme-real:compute-G}

    Run \cref{alg:exact-mme-integers} (\texttt{Exact-MME-integers}) with inputs $\inparen{G, \inparen{\widehat{\veca^{(1)}}, \ldots, \widehat{\veca^{(N)}}}, s = 3kdm}$ to obtain $B_1,\ldots, B_N$ such that, for all $i \in [N]$, we have
    \[
        B_i = G(\widehat{\veca^{(i)}}).
    \]
    Let $b_i = \round{\sfrac{B_i}{2^{k_1 + k_2 d m - t}}}{}$ for each $i \in [N]$. \label{alg:approx-mme-real:call-mme-integers}

    \Return{$\inparen{b_1,\ldots, b_N}$.}
\end{algorithm}

\paragraph{Proof of correctness:} 
Without loss of generality, we may assume that $d,m$ are growing parameters (from \crefrange{alg:approx-mme-real:if-kronecker}{alg:approx-mme-real:kronecker}).

Note that for any $\veca^{(i)}$, we have
\begin{align*}
\abs{f(\veca^{(i)}) - \round{f}{k_1}(\round{\veca^{(i)}}{k_2})} & \leq \abs{f(\veca^{(i)}) - \round{f}{k_1}(\veca^{(i)})} + \abs{\round{f}{k_1}(\veca^{(i)}) - \round{f}{k_1}(\round{\veca^{(i)}}{k_2})}\\
\leq \sfrac{1}{2^{t+2}} + \sfrac{1}{2^{t+2}} \leq \sfrac{1}{2^{t+1}}.
\end{align*}
\SGnote{Not sure why cref is unable to combine}
where the last inequality uses \cref{obs:error-rounding-coeffs} and \cref{obs:error-rounding-points} with our choice of $k_1$ and $k_2$. Thus, it suffices to compute $\round{f}{k_1}\inparen{\round{\veca^{(i)}}{k_2}}$ for each $i \in [N]$. The polynomial $\round{f}{k_1}$ is computed in \cref{alg:approx-mme-real:round-f} and the points $\round{\veca^{(i)}}{k_2}$ are computed in \crefrange{alg:approx-mme-real:for-round-points}{alg:approx-mme-real:round-points}. Let $\widehat{\veca^{(i)}} = 2^{k_2} \round{\veca^{(i)}}{k_2} \in (-2^{k_2}, 2^{k_2})^m$. 

\noindent
Since each coefficient of $2^{k_1} \cdot \round{f}{k_1}$ is bounded in magnitude by $2^{k_1}$, we have
\begin{align*}
    \abs{G(\widehat{\veca^{(i)}})} & = \abs{\sum_{\vece \in \set{0,\ldots,d-1}^m} g_{\vece, k_1} \cdot 2^{(k_2 d m) - k_2\abs{\vece}} \cdot \widehat{\veca^{(i)}}^\vece} \leq d^m \cdot 2^{k_1} \cdot 2^{k_2 dm} \cdot 2^{k_2 d m} \leq 2^{3kdm}. 
\end{align*}
From the definition of $G(x_1,\ldots, x_m)$, note that
\begin{align*}
    G(\widehat{\veca^{(i)}}) & = \sum_{\vece \in \set{0,\ldots,d-1}^m} g_{\vece, k_1} \cdot 2^{(k_2 d m) - k_2\abs{\vece}} \cdot \widehat{\veca^{(i)}}^\vece\\
    & = \sum_{\vece \in \set{0,\ldots,d-1}^m} g_{\vece, k_1} \cdot 2^{(k_2 d  m)} \cdot \inparen{\sfrac{1}{2^{k_2}} \cdot \widehat{\veca^{(i)}}}^\vece\\
    & = 2^{(k_2 \cdot d \cdot m)} \cdot \sum_{\vece \in \set{0,\ldots,d-1}^m} g_{\vece, k_1} \cdot \round{\veca^{(i)}}{k_2}^\vece\\
    & = 2^{k_1 + k_2 d m} \cdot \round{f}{k_1}(\round{\veca^{(i)}}{k_2}).
\end{align*}
Since \cref{thm:exact-mme-over-integers} correctly computes the evaluations of $G(\vecx)$ on $\widehat{\veca^{(i)}}$'s, we have we have for each $i \in [N]$
\[
    \sfrac{1}{2^{k_1 + k_2dm}} \cdot G\inparen{\widehat{\veca^{(i)}}} = \round{f}{k_1}(\round{\veca^{(i)}}{k_2}) = \sfrac{B_i}{2^{k_1 + k_2 d m}}. 
\]
Finally, if $b_i = \round{\sfrac{B_i}{2^{k_1 + k_2 d m - t}}}{}$, then 
    $$\abs{\sfrac{b_i}{2^t} - \sfrac{B_i}{2^{k_1 + k_2 d m}}} = \sfrac{1}{2^t} \cdot \abs{b_i - \sfrac{B_i}{2^{k_1 + k_2 d m - t}}} \leq \sfrac{1}{2^{t+1}}.$$
\noindent
Hence, $$ \abs{f(\veca^{(i)}) - \sfrac{b_i}{2^t}} \leq \abs{f(\veca^{(i)}) - \round{f}{k_1}(\round{\veca^{(i)}}{k_2})} +  \abs{\round{f}{k_1}(\round{\veca^{(i)}}{k_2}) - \sfrac{b_i}{2^t}} \leq \sfrac{1}{2^t}.$$

\paragraph{Running time analysis:}

After \crefrange{alg:approx-mme-real:if-kronecker}{alg:approx-mme-real:kronecker}, we may assume that $d,m = \omega(1)$ and all coefficients of $f$ and coordinates of points are provided via approximation oracles with running time $\tilde{O}(r + m\log d)$ to compute an $r$-bit approximation. 

\crefrange{alg:approx-mme-real:round-f}{alg:approx-mme-real:round-points} overall takes time 
\[
    (d^m + Nm) \cdot \tilde{O}(k + O(m\log d)) = (d^m + Nm) \cdot \tilde{O}(t + O(m\log d)) = ((d^m + Nm) \cdot t)^{1 + o(1)}.
\]
Computing the coefficients of $G(\vecx)$ takes time $(d^m) \cdot \tilde{O}(kdm)$. By \cref{thm:exact-mme-over-integers}, \cref{alg:approx-mme-real:call-mme-integers} takes time 
\[
    ((d^m + Nm) \cdot 3kdm)^{1 + o(1)} = ((d^m + Nm) \cdot t)^{1 + o(1)}.
\]
Therefore, \cref{alg:approx-mme-real} takes $((d^m + Nm) \cdot t)^{1 + o(1)}$ overall. \\

\noindent
This completes the proof of \cref{thm:approx-MME-reals}. \qed

\section{Exact-MME over rationals with known output complexity}\label{sec:exact mme rationals}
We now use our algorithm for approximate-MME over real numbers to obtain a fast algorithm for exact-MME over the field of rational numbers. We start by formally stating the precise problem that we solve and then build upon some necessary preliminaries that we need for our algorithm. 

\subsection{The problem statement}

\paragraph{Input:}A polynomial $f(x_1,\ldots, x_m) \in \Q_{(-1,1)}[x_1,\ldots, x_m]$ of individual degree less than $d$, given as a list of $d^m$, a list of points $\veca^{(1)}, \ldots, \veca^{(N)} \in \Q_{(-1,1)}^m$, an integer parameter $s > 0$ such that all rational numbers in the coefficients of $f$, the coordinates of points and evaluations $f(\veca^{(i)})$ are expressible as rational numbers of the form $\sfrac{p}{q}$ with $\abs{p}, \abs{q} < 2^s$. 

\paragraph{Output:} Integers $b_1,\ldots, b_N, c_1,\ldots, c_N$ such that $f(\veca^{(i)}) = \sfrac{b_i}{c_i}$ for all $i \in [N]$. 

\begin{theorem}[Exact-MME over rationals] \label{thm:exact-mme-rationals}
    There is a deterministic algorithm (namely \cref{alg:exact-mme-rationals}) that on input as mentioned above returns the required output as mentioned above and runs in time $((d^m + Nm) \cdot s)^{1 + o(1)}$ for all $m \in \N$ and sufficiently large $d \in \N$. 
\end{theorem}

\paragraph{Main idea:} The main idea would be a reduction to approximate-MME (\cref{thm:approx-MME-reals}) followed by a \emph{rational number reconstruction} step. If we can compute $f(\veca^{(i)})$ to a reasonable degree of accuracy (depending on the output guarantee $s$), we can recover the rational number exactly from it. Before we present the algorithm for the above theorem, we discuss the notion of continued fractions which would be the key to reconstructing the rational number of interest. 

\subsection{Continued fractions, rational approximations, and extended Euclid's algorithm}

\begin{definition}[Continued fractions] 
    A \emph{finite continued fraction} expressed by a sequence of integers $[q_1,\ldots, q_t]$ computes the rational number
    \[
        q_1 + \frac{1}{q_2 + \frac{1}{\ddots + \frac{1}{q_{t-1} + \frac{1}{q_t}}}}.
    \]

    An \emph{infinite continued fraction} expressed by an infinite sequence of integers $[q_1,q_2,\ldots]$ satisfying\footnote{Traditionally, continued fractions with this condition are called `simple' continued fractions but we will drop this qualifier as we will only deal with continued fractions with this additional constraint.} $q_2,\ldots, q_n > 0$  is said to compute a real number $\alpha$ if
    \[
        \alpha = q_1 + \frac{1}{q_2 + \frac{1}{q_3 + \frac{1}{\ddots}}}.
    \]
    in the sense that $\lim_{n\rightarrow \infty} \abs{\alpha - [q_1,\ldots, q_n]} = 0$. 
\end{definition}

We note some basic properties of continued fractions which may be found in most standard texts (cf. Schmidt~\cite[Chapter 1]{Schmidt-diophantine}).

\begin{proposition}[Uniqueness of continued fractions (Lemma 4C, 4D in \cite{Schmidt-diophantine})]
    \label{prop:continued-fraction-unique}
    Every real number has a unique continued fraction expansion up to the following exceptions:
    \begin{enumerate}
        \item \label{item:continued-fraction-unique:integer} If $\alpha$ is an integer, then there are exactly two continued fraction representations for $\alpha$ namely $[\alpha]$ and $[\alpha-1, 1]$.
        \item \label{item:continued-fraction-unique:rational} If $\alpha$ is a non-integral rational number, then there are exactly two continued fraction representations for $\alpha$: one of the form $[q_1,\ldots, q_n]$ with $q_n \geq 2$, and $[q_1,\ldots, q_{n}-1, 1]$ being the other. 
        \item If $\alpha$ is irrational, then there is exactly one continued fraction representation for $\alpha$.
    \end{enumerate}
\end{proposition}

\begin{definition}[Convergents]
    For a real number $\alpha$ with $[q_1,q_2,\ldots]$ being the unique\footnote{As a convention, for rational numbers, we will only consider continued fraction representations of the first form described in \cref{prop:continued-fraction-unique}~\cref{item:continued-fraction-unique:integer,item:continued-fraction-unique:rational}.}, the rational number $\sfrac{a_i}{b_i}$ corresponding to the $i$-th prefix $[q_1,\ldots, q_i]$ is called \emph{the $i$-th convergent of $\alpha$}. 
\end{definition}

\begin{lemma}[Properties of convergents]\label{lem:convergent-properties}
    Suppose $\set{\sfrac{a_i}{b_i}}_i$ be the convergents of a real number $\alpha = [q_1,q_2,\ldots]$. Then
    \begin{enumerate}
        \item \label{lem:convergent-properties:denominator} For any $n \geq 3$, we have
        \begin{align*}
            a_n & = q_n a_{n-1} + a_{n-2},\\
            b_n & = q_n b_{n-1} + b_{n-2}.
        \end{align*}
        In particular, the denominator sequence $\set{b_n}_{n\geq 2}$ is increasing.
        \item \label{lem:convergent-properties:successive-error} For all $n\geq 1$, 
        \begin{align*}
            \frac{a_{n+1}}{b_{n+1}} - \frac{a_n}{b_n} = \frac{(-1)^{n-1}}{b_n (q_{n+1} b_n + b_{n-1})} = \frac{(-1)^{n-1}}{b_n b_{n+1}}.
        \end{align*}
        \item \label{lem:convergent-properties:error} For any $n \geq 1$, unless $\alpha = \frac{a_n}{b_n}$, we have
        \[
            \frac{1}{b_n (b_{n} + b_{n+1})} \leq \abs{\alpha - \frac{a_n}{b_n}} \leq \frac{1}{b_n b_{n+1}}.
        \]
        \item \label{lem:convergent-properties:good-approx-implies-conv} Suppose $\sfrac{a}{b}$ is a rational number satisfying $\abs{\alpha - \sfrac{a}{b}} < \sfrac{1}{2b^2}$
        Then, $\sfrac{a}{b}$ is one of the convergents of $\alpha$.
    \end{enumerate}
\end{lemma}
\begin{proof}
    \cref{lem:convergent-properties:denominator,lem:convergent-properties:successive-error,lem:convergent-properties:good-approx-implies-conv} are just \cite[Lemma 3A, Lemma 3E, Theorem 5C]{Schmidt-diophantine} respectively. 

    For \cref{lem:convergent-properties:error}, if $\alpha \neq \sfrac{a_n}{b_n}$, we have that $q_{n+1}$ exists. Let $\alpha_{n+1} = [q_{n+1}, \ldots]$. Then, we may abuse notation and express $\alpha$ as the ``continued fraction'' $\alpha = [q_1, \ldots, q_n, \alpha_{n+1}]$. \cref{lem:convergent-properties:successive-error} for this expression yields
    \begin{align*}
        \abs{\alpha - \frac{a_n}{b_n}} & = \frac{1}{b_n (\alpha_{n+1} b_n + b_{n-1})}.
    \end{align*}
    Note that $q_{n+1} \leq \alpha_{n+1} \leq q_{n+1} + 1$ and hence 
    \begin{align*}
        \abs{\alpha - \frac{a_n}{b_n}} & = \frac{1}{b_n (\alpha_{n+1} b_n + b_{n-1})}\\
        & \leq\frac{1}{b_n (q_{n+1} b_n + b_{n-1})} = \frac{1}{b_n b_{n+1}}, \quad\text{(by \cref{lem:convergent-properties:denominator})}\\
        \text{and} \quad \abs{\alpha - \frac{a_n}{b_n}} & = \frac{1}{b_n (\alpha_{n+1} b_n + b_{n-1})}\\
        & \geq\frac{1}{b_n (q_{n+1} b_n + b_{n-1} + b_n)} = \frac{1}{b_n (b_n + b_{n+1})}.\qedhere
    \end{align*}
\end{proof}

\subsubsection*{Extended Euclid's Algorithm}

Closely related to continued fractions is the classical Extended Euclid's Algorithm for computing the greatest common divisor of two numbers. 

\begin{definition}[Remainder and quotient sequences] \label{defn:remainder-quotient-sequences}
    For a pair of integers $a, b > 0$, we define the \emph{remainder sequence} $\set{r_i}_{i = 0, \ldots, t+1}$ and the \emph{quotient sequence} $\set{q_i}_{i = 1, \ldots, t}$ for the pair $(a,b)$ as follows:
    \begin{itemize}\itemsep 0pt
        \item $r_0 = a$ and $r_1 = b$,
        \item For all $i \geq 1$, define $q_i, r_{i+1}$ as the quotient and remainder respectively when $r_{i-1}$ is divided by $r_i$. Thus,
        \[
            r_{i+1} = r_{i-1} \bmod{r_i} = r_{i-1} - q_i r_{i}.
        \]
        \item $r_{t+1}$ is the first element of the sequence that is equal to zero. \qedhere
    \end{itemize}
\end{definition}

\begin{observationwp}[Continued fractions for a rational number and quotient sequences]\label{obs:continued-fractions-eea}
    Suppose $a,b > 0$ are a pair of integers and $\set{q_1,\ldots, q_t}$ is the associated quotient sequence. Then, the continued fraction representation of the rational number $\sfrac{a}{b}$ is $[q_1,\ldots, q_t]$:
    \[
        \frac{a}{b} = q_1 + \frac{1}{q_2 + \frac{1}{\ddots + \frac{1}{q_t}}}.\qedhere
    \]
\end{observationwp}

Computing the $\gcd$ of two given integers, and more generally computing the entire quotient sequence, can be done in deterministic nearly-linear time; this is attributed to Knuth and \Schonhage{} (cf. M\"oller~\cite{Moller2008} for a complete description and a detailed history).

\begin{theorem}[Fast Extended Euclid Algorithm (cf. M\"oller~\cite{Moller2008} )] \label{thm:fast-eea}
    There is a deterministic algorithm that, on input a pair of integers $a > b > 0$ with $a,b \leq 2^s$, computes the entire quotient sequence $q_1,\ldots, q_t$ for the pair $(a,b)$ in time $\tilde{O}(s)$. 
\end{theorem}

\begin{corollary}[Fast computation of convergents]\label{cor:fast-computation-convergents}
    There is a deterministic algorithm that, on input a pair of integers $M, N > 0$ with $M,N \leq 2^s$, and an integer $i > 0$ computes integers $a_i, b_i$ such that $\sfrac{a_i}{b_i}$ is the $i$-th convergent of the rational number $\sfrac{M}{N}$, with running time $\tilde{O}(s)$.
\end{corollary}
\begin{proof}
    Let $q_1,\ldots, q_t$ be the quotient sequence for the pair $(M,N)$, which may be computed using \cref{thm:fast-eea} in $\tilde{O}(s)$ time. By \cref{obs:continued-fractions-eea}, this is the continued fraction representation of $\sfrac{M}{N}$. Thus, it is easy to note that
    \[
        \begin{bmatrix}
            a_{i} & a_{i-1}\\
            b_{i} & b_{i-1}
        \end{bmatrix} = \begin{bmatrix} q_1 & 1 \\ 1 & 0 \end{bmatrix} \cdots \begin{bmatrix} q_i & 1 \\ 1 & 0 \end{bmatrix}
    \]
    where $\sfrac{a_j}{b_j}$ is the $j$-th convergent. 
    Note that $\abs{q_j} < \sfrac{r_{j-1}}{r_j}$ where $\set{r_0,\ldots, r_t}$ is the associated remainder sequence and hence we have $\abs{q_1 \cdots q_t} \leq M \leq 2^s$. Thus, this matrix product can be computed in $\tilde{O}(s)$ time.
\end{proof}

\subsection{Rational number reconstruction}

\begin{lemma}[Fast rational number reconstruction]\label{lem:fast-rnr}
    There is a deterministic algorithm (namely \cref{alg:fast-rnr}) that, given as input an integer parameter $s > 0$ and integers $A,B$ with the guarantee that $\abs{B} < 2^{2s + 1}$ and there exist a unique rational number (in reduced form) $\sfrac{a}{b}$ with $\abs{b} < 2^s$ and
    \[
        \abs{\frac{A}{B} - \frac{a}{b}} < \frac{1}{2^{2s + 1}},
    \]
    finds the integers $a,b$ in time $\tilde{O}(s)$. 
\end{lemma}
\begin{proof}
    The algorithm is straightforward given \cref{cor:fast-computation-convergents} and \cref{lem:convergent-properties}.

    \medskip

    \begin{algorithm}[H]
        \caption{\texttt{Fast-Rational-Number-Reconstruction}}
        \label{alg:fast-rnr}
        \SetKwInOut{Input}{Input}\SetKwInOut{Output}{Output}
    
        \Input{Integers $A,B$ and an integer parameter $s > 0$ such that $\abs{A}, \abs{B} \leq 2^{2s + 1}$ and there is some rational number $\sfrac{a}{b}$ such that $\abs{b} < 2^s$ and $\abs{\sfrac{A}{B}- \sfrac{a}{b}} < \sfrac{1}{2^{2s+1}}$.}
        \Output{The integers $a,b$.}
        \BlankLine

        Using \cref{thm:fast-eea}, compute the quotient sequence $q_1,\ldots, q_\ell$ for the pair $A, B$.

        Using \cref{cor:fast-computation-convergents} and binary search, compute the largest index $i$ such that the $i$-th convergent $\sfrac{a_i}{b_i}$ satisfies $\abs{b_i} < 2^s$. \label{alg:fast-rnr:binary-search}

        \Return{$a_i, b_i$.}
    \end{algorithm}

    \medskip

    The running time of the algorithm is clearly $\tilde{O}(s)$ as claimed as $\ell = O(\log (A+B)) = O(s)$ and thus we have at most $O(\log \ell) = O(\log s)$ uses of \cref{cor:fast-computation-convergents} in \cref{alg:fast-rnr:binary-search}.

    For correctness, assume that $\sfrac{A}{B}$ is in its reduced form. Since we know $b_1 = 1$, let $i$ be the largest index with the denominator $b_i$ of the convergent $\sfrac{a_i}{b_i}$ satisfies $b_i < 2^s$. If this is the last convergent, then $\sfrac{A}{B} = \sfrac{a_i}{b_i}$ and we are done. Thus, we may assume that $\sfrac{A}{B} \neq \sfrac{a_i}{b_i}$. 

    Since we are given that $\abs{\sfrac{A}{B} - \sfrac{a}{b}} < \sfrac{1}{2^{2s+1}} < \sfrac{1}{2b^2}$, by \cref{lem:convergent-properties} \cref{lem:convergent-properties:good-approx-implies-conv}, $\sfrac{a}{b}$ is one of the convergents of $\sfrac{A}{B}$. For any $\ell > i$, the $\ell$-th convergent $\sfrac{a_\ell}{b_\ell}$ has denominator larger than $2^s$. For any $j < i$, from \cref{lem:convergent-properties} \cref{lem:convergent-properties:error} and \cref{lem:convergent-properties:denominator} we have
    \[
        \abs{\frac{A}{B} - \frac{a_j}{b_j}} \geq \frac{1}{b_j (b_{j} + b_{j+1})} > \frac{1}{2\cdot b_i^2} \geq \frac{1}{2^{2s+1}}.
    \]
    Thus, $\sfrac{a}{b}$ must be the $i$-th convergent $\sfrac{a_i}{b_i}$. 
\end{proof}

\subsection{Algorithm for exact-MME over rationals}

We now have all the necessary ingredients to describe the algorithm to prove \cref{thm:exact-mme-rationals}. 

\medskip

\begin{algorithm}[H]
    \caption{\texttt{Exact-MME-Rationals}}
    \label{alg:exact-mme-rationals}
    \SetKwInOut{Input}{Input}\SetKwInOut{Output}{Output}

    \Input{A polynomial $f(x_1,\ldots, x_m) \in \Q_{(-1,1)}[\vecx]$, points $\veca^{(1)}, \ldots, \veca^{(N)} \in \Q_{(-1,1)}^m$, with all rational numbers provided via the numerator and denominator, and an integer parameter $s$ such that all numerators and denominators of the coefficients of $f$, coordinates of the points, and evaluations $f(\veca^{(i)})$ are at most $2^s$.}
    \Output{Integers $b_1,\ldots, b_N$ and $c_1,\ldots, c_N$ such that $f(\veca^{(i)}) = \sfrac{b_i}{c_i}$ for all $i\in [N]$.}
    \BlankLine

    Using the numerators and denominators for the required approximation oracles, run $\texttt{approximate-MME-Reals}\inparen{f, \set{\veca^{(1)}, \ldots, \veca^{(N)}}, t = 2s+1}$ (\cref{alg:approx-mme-real}) to obtain integers $(B_1,\ldots, B_N)$ such that 
    \[
        \abs{f(\veca^{(i)}) - \frac{B_i}{2^t}} < \frac{1}{2^{t}} = \frac{1}{2^{2s+1}}. 
    \]    \label{alg:exact-mme-rationals:run-approx-mme}

    \For{$i\in N$}{\label{alg:exact-mme-rationals:for-rnr}
        Run $\texttt{Fast-Rational-Number-Reconstruction}\inparen{B_i, 2^{2s+1}, s}$ (\cref{alg:fast-rnr}) to get $b_i,c_i$ with $\abs{c_i} < 2^s$ such that 
        \[
            \abs{\frac{B_i}{2^{2s+1}} - \frac{b_i}{c_i}} < \frac{1}{2^{2s+1}}.
        \]\label{alg:exact-mme-rationals:rnr}
    }

    \Return{$\inparen{b_1,\ldots, b_N}, \inparen{c_1,\ldots, c_N}.$}
\end{algorithm}

\medskip

The correctness of \cref{alg:exact-mme-rationals} is evident from \cref{thm:approx-MME-reals} and \cref{lem:fast-rnr}. \cref{thm:approx-MME-reals} asserts that \cref{alg:approx-mme-real} correctly provides the required approximations for the  evaluations, and \cref{lem:fast-rnr} asserts that \cref{alg:fast-rnr} reconstructs the correct rational number. 

As for running time, given the numerator and denominators, we can build approximation oracles for each rational number with nearly-linear running time. Thus, \cref{alg:exact-mme-rationals:run-approx-mme} takes $((d^m + Nm) \cdot s)^{1 + o(1)}$ time and the loop in \crefrange{alg:exact-mme-rationals:for-rnr}{alg:exact-mme-rationals:rnr} takes $\tilde{O}(N \cdot s)$ time. Thus, the total running time is $((d^m + Nm) \cdot s)^{1 + o(1)}$ as claimed. \\

This completes the proof of \cref{thm:exact-mme-rationals}. \qed

\section{Approximate-MME over complex numbers}
In this section, we briefly discuss the extension of \cref{thm:approx-MME-reals} to the field of complex numbers. As discussed in the preliminaries, the field constants in this case are given by two approximation oracles, one for the real part of the complex number, and one for the imaginary part. The ideas needed for this extension, on top of the ideas in the proof of \cref{thm:MME-finite-fields} are quite standard and were introduced by Kedlaya \& Umans \cite{KedlayaU2011} for designing fast algorithms for MME for finite fields that are not prime. This approach also found a subsequent application in the work of Bhargava et al. \cite{BhargavaGGKU2022}, again in the context of dealing with non-prime finite fields while designing algorithms for MME. In the interest of keeping this discussion succinct and to avoid repetition, we outline the main steps needed for this generalization, but skip the formal details. The structure of the algorithm closely follows that of \cref{alg:approx-mme-real}, with some additional care.  

As in the proof of \cref{alg:approx-mme-real}, we first make sure that the number of underlying variables is growing. Next, we round each of the field constants (both the real and the imaginary parts) by rational numbers with denominator $2^k$ for some sufficiently large integer $k$ to be chosen later. At this point, we have introduced some error (which turns out to be small if $k$ is sufficiently large), but have reduced the problem instance over $\C$ to an instance over $\Q[\omega]$, where $\omega$ is a square root of $-1$. Moreover, all the denominators of the field constants in the problem are of the form $2^k$. We now clear out the denominators, as in \cref{alg:approx-mme-real}, and get  an instance of MME where the constants in the problem are from the ring $\Z[\omega]$. At this point, we replace $\omega$ in the constants in the input by a new formal variable $z$, and instead of working over the ring $\Z[\omega]$, we work over the ring $\Z[z]/\langle z^2 + 1 \rangle$. Note that this is sufficient, since given a solution to MME over this ring, we can obtain a solution to the original problem by just replacing $z$ by $\omega$. Now, the idea is to just invoke the algorithm for exact MME over integers (\cref{alg:exact-mme-integers}) for this problem instance. However, we cannot quite do that directly since the instance at hand is over $\Z[z]/\langle z^2 + 1 \rangle$ and not over $\Z$ as desired. Nevertheless, we proceed as in \cref{alg:exact-mme-integers} by picking sufficiently many primes $p_1, p_2, \ldots, p_s$ and reducing the problem instance over $\Z[z]/\langle z^2 + 1 \rangle$ modulo these primes to obtain instances over the rings $\F_{p_i}[z]/\langle z^2 + 1 \rangle$ for every $i$. In \cref{alg:exact-mme-integers}, we just invoked the result of \cite{BhargavaGGKU2022} over prime fields at this stage, and then combined the outputs using fast Chinese remaindering. However, in this case, what we have are instances over the finite rings $\F_{p_i}[z]/\langle z^2 + 1 \rangle$. But this does not turn out to be an issue as  the algorithm of Bhargava et al continues to work over such rings, and indeed the results and proofs in the \cite{BhargavaGGKU2022} are stated in this form. One final thing to note is that the small optimizations that we do over the results in \cite{BhargavaGGKU2022} in \cref{sec:optimization over finite fields} to make sure that the dependence of the running time on the field size is nearly-linear continues to be true for the extension rings that we have here. Once we have solved all the instances over $\F_{p_i}[z]/\langle z^2 + 1 \rangle$, we can recover the solution over $\Z[z]/\langle z^2 + 1 \rangle$ by an application of fast Chinese remaindering as in \cref{alg:exact-mme-integers}, and an appropriate scaling of these evaluations (again, as in \cref{alg:approx-mme-real}) gives us approximations of the original evaluations over $\C$. The error analysis and the bound on the running time essentially the same as that in the analysis of \cref{alg:approx-mme-real}. We skip the rest of the details.

\section{Discussion and open problems} \label{sec: open problems}
We conclude with some open problems. 

\begin{enumerate}
    \item Perhaps the most natural question here is to seek an algebraic algorithm for multivariate multipoint evaluation over general fields, both finite and infinite. Currently, we only know such algebraic algorithms over finite fields of small characteristic \cite{Umans2008, BhargavaGKM2022}. 
    \item The aforementioned question of having an algebraic algorithm for MME is also interesting in the non-uniform setting. For instance, we do not know if the linear transformation given by a multivariate Vandermonde matrices can be computed by an arithmetic circuit of nearly-linear (or even sub-quadratic) size  over fields other than finite fields of small characteristic. 
    \item It would be interesting to have additional applications of these faster algorithms and the ideas therein, beyond the applications already mentioned by Kedlaya and Umans \cite{KedlayaU2011}. 
\end{enumerate}

{\small 

  \bibliographystyle{prahladhurl}
  \bibliography{MME-bib.bib}
}

\end{document}